\documentclass{article}
\usepackage[totalwidth=460pt, totalheight=680pt]{geometry}
\usepackage[hyphens]{url}  %
\usepackage{graphicx} %
\usepackage{natbib}  %
\usepackage{caption} %
\DeclareCaptionStyle{ruled}{labelfont=normalfont,labelsep=colon,strut=off} %
\usepackage[noend]{algorithm2e}

\usepackage{thmtools}
\usepackage{thm-restate}
\usepackage{xspace}
\usepackage{multirow}
\usepackage{mathtools}
\usepackage{amsfonts,amsmath,amsthm}
\newtheorem{theorem}{Theorem}

\newtheorem{lemma}{Lemma}
\newtheorem{observation}{Observation}
\newtheorem{definition}{Definition}

\usepackage{subfig}
\usepackage{comment}
\usepackage{tikz}
\usetikzlibrary{backgrounds,positioning,calc,shapes,fit,matrix}
\usetikzlibrary{arrows,decorations.pathmorphing,decorations.pathreplacing}
\usepackage{mdframed}
\usepackage{pgfplots}
\pgfplotsset{compat=1.9}
\usepgfplotslibrary{groupplots}
\usepackage{enumitem}

\usepackage[textwidth=.6\marginparwidth,obeyFinal,colorinlistoftodos,textsize=scriptsize]{todonotes}

\newcommand{\mytodo}[2]{\todo[size=\tiny, color=#1!50!white]{#2}\xspace}
\newcommand{\myrevtodo}[2]{{%
    \let\marginpar\marginnote
    \reversemarginpar
    \renewcommand{\baselinestretch}{0.8}%
    \todo[size=\tiny, color=#1!50!white]{#2}}}
\newcommand{\myinlinetodo}[2]{\todo[size=\small, color=#1!50!white, inline, 
caption={}]{#2}\xspace}
\newcommand{\registerAuthor}[3]{%
  \expandafter\newcommand\csname #2com\endcsname[1]{\mytodo{#3}{\textsc{#2}: 
  ##1}}%
  \expandafter\newcommand\csname 
  #2revcom\endcsname[1]{\myrevtodo{#3}{\textsc{#2}: ##1}}%
  \expandafter\newcommand\csname 
  #2inline\endcsname[1]{\myinlinetodo{#3}{\textsc{#2}: ##1}}%
  \expandafter\newcommand\csname 
  #2inlineLater\endcsname[1]{\lv{\myinlinetodo{#3}{\textsc{#2}: ##1}}}%
}
\registerAuthor{Robert bredereck}{rb}{yellow}
\registerAuthor{Andrzej kaczmarczyk}{ak}{green}

\usepackage{nicefrac}
\usepackage{booktabs}
\usepackage{multirow}
\usepackage{dsfont}
\usepackage{diagbox}
\usepackage{etoolbox}
\usepackage{appendix}
\usepackage{makecell}
\usepackage{mathtools}

\newcommand{\orderedlistingof}[2]{\ensuremath{#1_1, #1_2, \ldots, #1_{#2}}}
\newcommand{\orderedsetof}[2]{\ensuremath{\{\orderedlistingof{#1}{#2}\}}}

\newcommand{\namedorderedsetof}[3]{\ensuremath{{#1}=\orderedsetof{#2}{#3}}}

\newcommand{\clique}{\textsc{Clique}}

\newcommand{\threeCNFSAT}{\textsc{3-CNF-Sat}}

\newcommand{\np}{\ensuremath{\mathrm{NP}}}

\newcommand{\fpt}{\ensuremath{\mathrm{FPT}}}                                     
\newcommand{\xp}{\ensuremath{\mathrm{XP}}}                                       
\newcommand{\wone}{\ensuremath{\mathrm{W[1]}}}                                   
\newcommand{\wonehard}{\wone-hard}
\newcommand{\woneh}{\wone-h}
\newcommand{\wonehardness}{\wonehard{}ness}

\newcommand{\p}{\ensuremath{\mathrm{P}}}

\newcommand{\nphard}{\np-hard}                                       
\newcommand{\nph}{\np-h}
\newcommand{\pnph}{p-\np-h}

\newcommand{\naturals}{\mathbb{N}}

\newenvironment{problemQuote}%
  {\list{}{\leftmargin=0.11in\rightmargin=0.11in}\item[]}%
  {\endlist}

\newcommand{\problemName}{\textsc{Sharing Allocation}}
\newcommand{\problemNameShort}{\textsc{SA}}

\newcommand{\utilProblemShort}{\textsc{UW}\problemNameShort}
\newcommand{\egalProblemShort}{\textsc{EW}\problemNameShort}
\newcommand{\envyReducingProblem}{\textsc{Envy-Reducing} \problemName}
\newcommand{\envyReducingProblemShort}{\textsc{ER}\problemNameShort}
\newcommand{\MCC}{\textsc{Multicolored Clique}}
\newcommand{\IS}{\textsc{Independent Set}}
\newcommand{\SBMWM}{\textsc{Size-bounded Maximum Weighted Matching}{}}
\newcommand{\WBMM}{\textsc{Weight-bounded Maximum Matching}{}}
\newcommand{\SBMWMShort}{\textsc{SBMWM}{}}
\newcommand{\WBMMShort}{\textsc{WBMM}{}}

\DeclareMathOperator{\Env}{Env}

\newcommand{\allocation}{\ensuremath{\pi}}
\newcommand{\resources}{\ensuremath{\mathcal{R}}}
\newcommand{\agents}{\ensuremath{\mathcal{A}}}
\newcommand{\resource}{\ensuremath{r}}
\newcommand{\agent}{\ensuremath{a}}
\newcommand{\resourcesCount}{\ensuremath{m}}
\newcommand{\agentsCount}{\ensuremath{n}}
\newcommand{\utilityFunction}{\ensuremath{u}}
\DeclareMathOperator{\utilitarianWelfare}{usw}
\DeclareMathOperator{\egalitarianWelfare}{esw}
\newcommand{\graph}{\ensuremath{G}}
\newcommand{\attentionGraph}{\ensuremath{\mathcal{G}_t}}
\newcommand{\shareGraph}{\ensuremath{\mathcal{G}_s}}
\newcommand{\attentionRelations}{\ensuremath{\mathcal{E}_t}}
\newcommand{\shareRelations}{\ensuremath{\mathcal{E}_s}}
\newcommand{\sharing}{\ensuremath{\delta}}
\newcommand{\incidentTo}{\ensuremath{I}}
\newcommand{\sharingAllocation}{\ensuremath{\Pi}}

\usepackage{authblk}

\title{On Improving Resource Allocations by Sharing}
\author[1] {Robert Bredereck}
\author[2,3] {Andrzej Kaczmarczyk}
\author[4] {Junjie Luo}
\author[2] {Rolf Niedermeier}
\author[2] {Florian~Sachse}

\affil[1]{Humboldt-Universit\"at zu Berlin, Berlin, Germany}
\affil[2]{Algorithmics and Computational Complexity, TU Berlin, Berlin, Germany}
\affil[3]{AGH University, Krak\'ow, Poland}
\affil[4]{Nyanyang Technological University, Singapore}

\affil[ ]{\small\texttt{robert.bredereck@hu-berlin.de, andrzej.kaczmarczyk@agh.edu.pl, junjie.luo@ntu.edu.sg, rolf.niedermeier@tu-berlin.de,
sachse.florian@gmail.com}}

\begin{document}

\maketitle
\begin{abstract}
 Given an initial resource allocation, where
 some agents may envy others or where a different distribution of 
 resources might lead to higher social welfare, our goal is to improve the
 allocation \emph{without} reassigning resources. %
 We consider a sharing concept allowing resources
 being shared with social network neighbors of the resource owners.
 To this end, we introduce a formal model that allows a central authority
 to compute an optimal
 sharing between neighbors based on an initial allocation. Advocating this
 point of view, we focus on the most basic scenario where a
 resource may be shared by two neighbors in a social network 
 and each agent can participate in a bounded number of sharings.
 We present algorithms for
 optimizing utilitarian and egalitarian social welfare of allocations and
 for reducing the number of envious agents. In particular, we examine the
 computational complexity with respect to several natural parameters.
 Furthermore, we study cases with restricted 
 social network structures and, among
 others, devise polynomial-time algorithms in path- and tree-like
 (hierarchical) social networks.
\end{abstract}

\section{Introduction}
The fair allocation of resources undoubtedly is a key challenge for modern 
societies and economies. Applications can be found in such diverse 
fields as cloud computing, food banks, or managing carbon loads in the 
context of global warming. %
Naturally, this 
topic received high attention in the scientific literature.
This also holds true for the special case of indivisible resources~\citep{BCM16}, 
which we 
concentrate on here. Moreover, we take into account the role 
of social networks built by agents, a growing line of 
research~\citep{AKP17,BQZ17,BCEIP17,BKN18,CEM17,BCGHLMW19,DBLP:conf/aldt/LangeR19,HX19}.
We bring one further new aspect into this scenario, reflecting
the increasing relevance of ``sharing economies''~\citep{BEB19,SC19},
where agents share resources in 
a peer-to-peer fashion. Resources to share
may be almost everything, for instance, knowledge, machines, time, or 
natural resources. More specifically, sharing in our scenario, which takes into
account the relationships between agents expressed by social networks, 
means that \emph{two} adjacent 
agents in the social network may share the very same resource, thus 
increasing the utility of the resource allocation 
for at least one of them (assuming positive utility for each resource). 
We assume this to be organized and decided 
by a central authority like, for example, the boss of a company.
To get started with this new setting, we focus on a very basic 
scenario. That is, in our model only two neighbors may share
\emph{and}, reflecting the (very human) principle of protection 
of acquired possession, no agent shall loose its already allocated 
resources. This conservative principle naturally makes sharing 
easier to implement, keeping ``restructuring costs'' lower, 
and it may even help to ``keep peace'' among agents. 
Moreover, it sometimes comes very naturally as depicted in the subsequent 
knowledge sharing example.
Besides improving egalitarian or utilitarian welfare, we
focus on the perhaps most basic fairness criterion, envy-freeness. 
Since it is not always possible that complete envy-freeness is achieved
(consider one indivisible resource and two agents
desiring it),
we aim at, given an initial resource allocation, improving it
by decreasing the number of envious agents 
through resource sharing. Moreover, we allow for 
modeling relationship aspects of sharing based on the social network
formed by the agents.

Before becoming more specific about our model, let us first introduce the
following example related to knowledge sharing.
Assume that agents are employees of a company, each having a bundle of 
qualifications. An agent may ``envy''
another agent because the other agent has some special qualification. 
The central authority wants to improve the
situation
by building teams of two agents where, due to a daily extensive cooperation, one
teaches the other the missing qualification (for instance, a
realization of this %
is the concept of~\emph{pair programming}
that also has other benefits besides knowledge
sharing~\citep{WKCJ00}).

\paragraph*{Model of sharing allocation.}
Roughly speaking, our model %
is as follows (see Section~\ref{sec:prelim} for formal definitions).
The input is a set of agents and a set of indivisible resources initially 
assigned to the agents. 
Typically, every agent may be assigned several resources.
Each agent has an individual utility value for 
each resource. The general goals are to decrease 
the overall degree of envy, to increase the sum of ``utility scores'' of all
agents, or to increase the minimal ``utility score'' among all agents.
Importantly, the only way 
an agent can improve its individual ``utility score'' is by participating in a sharing with other agents.

We assume that if an owner shares, then this does not decrease its own overall
utility value. This approach is justified when the burden of sharing is
neutralized by its advantages. Indeed, in our knowledge sharing
example a hassle of cooperation is often compensated by a better working
experience or higher quality outcomes (as shown
by~\citet{WKCJ00}). Note that such complicated mutual dependencies that would
be extremely hard to describe formally form a natural field for our
approach.
Further application examples
include irregularly used resources (like printers or compute servers). Here, the
coordination %
with another person
is uncritical and splitting the maintenance costs neutralizes the inconvenience
of cooperation.

We enrich our model by using two graphs, 
an undirected sharing graph and a directed attention graph, 
to model social relations between agents and to govern the
following two constraints of our model. 
The sharing graph models the possibility for two agents to share resources because, e.g., they are close to each other or there is no conflict between the time they use resources.
We focus on the case when only neighbors in
the sharing graph can share a resource (a missing qualification in our knowledge
sharing example). 
With respect to lowering the degree of envy, we assume that
agents may only envy their outneighbors in the directed attention graph.
This graph-based envy concept has recently been studied by many works in fair allocation \cite{BKN18,ABCGL18,BCGHLMW19}.

Agents may naturally be conservative in the sense of keeping control and
not sharing too much. Furthermore, as in our initial example, it might simply
be too ineffective to share a qualification among more than two employees
simultaneously (due to, e.g.,\,increased communication overhead or additional
resources needed). 
We address this in the most basic way and assume that each resource can be shared
to \emph{at most one neighbor} of its owner and an agent can participate in a bounded
number of sharings. %
This strong restriction already leads to tricky algorithmic challenges and
fundamental insights. In particular, 
the model also naturally extends on well-known 
matching scenarios in a
non-trivial way. %

There are numerous options to
further extend and generalize our basic model, as discussed in Section \ref{sec:extension} and in the
concluding~Section~\ref{sec:concl}. However, keeping our primary model simple, we aim at
spotting its fundamental properties influencing the complexity of
related computational problems.

\paragraph*{Related work.}
To the best of our knowledge, so far the model we consider has not been studied.
Since obtaining envy-free allocations is not always possible, 
there has been work on relaxing the 
concept of envy.
In particular, in the literature 
\emph{bounded-maximum
envy}~\citep{LMMS04}, \emph{envy-freeness up to one good}~\citep{Bud11},
\emph{envy-freeness up to the least-valued good}~\citep{CKMPSW19},
\emph{epistemic envy-freeness}~\citep{ABCGL18}, and
\emph{maximin share guarantee}~\citep{Bud11} have been studied.
However, these concepts combat nonexistence of allocations that are envy-free by
considering approximate versions of it;
they basically do not try to
tackle the question of how to achieve less ``envy'' in an allocation.
By way of contrast, 
our approach tries to find a way to lessen envy not by relaxing
the concept of envy, but rather by enabling a small
deviation in the model of indivisible, non-shareable resources. 
To this end, we make resources
shareable (in our basic model by two agents).
This approach is in line with a series of recent works which try to reduce envy
(i) by introducing small amounts of
money~\citep{BDNSV20,Halpern_2019,caragiannis2020computing}, (ii) by
donating a small set of resources to charity~\citep{CGH19,CKMS21},
or (iii) by allowing dividing a small number of indivisible
resources~\citep{SS19,S19}.
In particular, the papers mentioned in point (iii)
consider a model of indivisible resources that could be shared by an
arbitrary group of agents and where, unlike in our study, each agent only gets a
portion of the utility of the shared resources. Contrary to our setting, this model assumes no
initial allocation.
As a result, an envy-free allocation always exists and the goal
is to seek one with a minimum number of shared resources. In contrast,
our goal is to improve an \emph{initial} allocation through sharing
resources between pairs of agents.
Another line of research considers the improvement of allocations
by exchanging items~\cite{CEM07,GLW17,HX19}.
There has been quite some work on bringing together resource allocation and
social
networks~\citep{AKP17,BQZ17,BCEIP17,BKN18,CEM17,BCGHLMW19,HX19}. In
particular, the concept of only local envy relations to neighbors in 
a graph gained quite some attention~\cite{ABCGL18,BCGHLMW19,BKN18,EGHO20}.
Modifying existing allocations to maintain fairness over time 
has also been studied in online settings
with changing agents~\citep{f15dynamic,f17dynamic} or
resources arriving over time~\citep{ijcai2019-49}.

\paragraph*{Our contributions.}
Introducing
a novel model for (indivisible) resource allocation with agents
linked by social networks, we provide a %
view on improving
existing allocations for
several measures without, conceivably impossible, reallocations.

We analyze the (parameterized) computational complexity of applying our model to
improve utilitarian social welfare or egalitarian social welfare (Definition~\ref{def:Pwelfare}),
and to decrease the number of envious agents (Definition~\ref{def:Penvy}).
We show that
a central authority can (mostly) find a sharing that improves social welfare
(measured in both the egalitarian and utilitarian ways) in polynomial time,
while decreasing the number of envious agents is \nphard{} even if the sharing
graph is a clique and the attention graph is a bidirectional clique. 
To overcome
\np-hardness, we also study the influence of different natural parameters (such
as agent utility function values, structural parameters concerning the agent social
networks, the number of agents, and the number of resources);
Table~\ref{tab:results_parameterized} surveys our results in more detail.
We show that the problem is polynomial-time solvable if the underlying undirected graph of the attention graph is the same as the sharing graph and has constant treewidth (close to a tree).
We also identify an interesting contrast between the roles of the two graphs:
When agents have the same utility function, the problem is solvable in polynomial time if the attention graph is a bidirectional clique, while the problem is \nphard{} even if the sharing graph is a clique.
Finally, we show that the problem is fixed-parameter
tractable (\fpt{}) for the parameter number of agents (giving hope for efficient
solutions in case of a small number of agents) and polynomial-time solvable
(in~\xp) for a constant number of resources. However, the problem is \nphard{}
even if the goal is to reduce the number of envious agents from one to zero.

Altogether, our main technical contributions are with respect to exploring
the potential to ``overcome'' the NP-hardness of decreasing the number of
envious agents by exploiting several problem-specific parameters.

\setlength{\tabcolsep}{10pt}
\begin{table*}[t]
 \centering
 \resizebox{\textwidth}{!}{%
 \begin{tabular}{ccccccccc}
 \multicolumn{9}{c}{\envyReducingProblem{} (\envyReducingProblemShort{})}             \\ \toprule
 \multicolumn{2}{c}{$\shareGraph=\attentionGraph$} & \multicolumn{2}{c}{same utility} & \multicolumn{2}{c}{few agents} &
 \multicolumn{3}{c}{few resources} \\ 
     \cmidrule(l{.5em}r{.5em}){1-2}
     \cmidrule(l{.5em}r{.5em}){3-4}
     \cmidrule(l{.5em}r{.5em}){5-6}
      \cmidrule(l{.5em}r{.5em}){7-9}
 \multirow{2}{*}{clique}  & 
  \multirow{2}{*}{\begin{tabular}[c]{@{}c@{}}tree- or \\ pathwidth\end{tabular}}  & 
   \multirow{2}{*}{\begin{tabular}[c]{@{}c@{}}$\attentionGraph=$  clique\end{tabular}} &
    \multirow{2}{*}{\begin{tabular}[c]{@{}c@{}}$\shareGraph=$  clique\end{tabular}} &
  \multirow{2}{*}{$n$} &
  \multirow{2}{*}{$k=0, \Delta k=1$} &
  \multirow{2}{*}{$m$} &
  \multirow{2}{*}{\begin{tabular}[c]{@{}c@{}}\#shared \\ resources\end{tabular}} &
  \multirow{2}{*}{\begin{tabular}[c]{@{}c@{}} bundle \\ size\end{tabular}}  \\ 
   & & & & & & & & \\
    \cmidrule(l{.5em}r{.5em}){1-2}
    \cmidrule(l{.5em}r{.5em}){3-4}
    \cmidrule(l{.5em}r{.5em}){5-6}
    \cmidrule(l{.5em}r{.5em}){7-9}

  \nph & \xp, \woneh & \p & \nph & \fpt & \pnph & \multicolumn{2}{c}{\xp,
  \woneh} & \pnph \\
 \small Thm.~\ref{thm:hard_clique} &
 \small Thm. \ref{thm:XP_treewidth} &
 \small Thm.~\ref{thm:P_clique+same_utility} &
 \small Thm.~\ref{thm:hard_share_clique+same_utility} &
 \small Thm.~\ref{thm:fpt-agents} &
 \small Thm.~\ref{thm:decreasing_by_one_nph} &
 \multicolumn{2}{c}{\small Obs.~\ref{obs:resources_envy_red_xp},
         Thm.~\ref{thm:hard_share_clique+same_utility}} &
 \small Thm.~\ref{thm:hard_share_clique+same_utility}
 \\ \bottomrule
 \end{tabular}%
 }
 \caption{
  Results overview for \envyReducingProblemShort{}, where $\shareGraph=\attentionGraph$ means that the sharing graph (\shareGraph) is the same as the underlying graph of the attention graph (\attentionGraph), $n$~is the number of
  agents, $k$ is the number of envious agents after sharing, $\Delta k$ is a drop in the
  number of envious agents, and $m$ is the number of resources.
 }\label{tab:results_parameterized}%
\end{table*}%

\section{Preliminaries}\label{sec:prelim}

For a set~\namedorderedsetof{\agents}{\agent}{\agentsCount} of \emph{agents} and
a set~\namedorderedsetof{\resources}{\resource}{\resourcesCount}
of~\emph{indivisible resources}, a \emph{(simple) allocation}~$\allocation
\colon \agents \rightarrow 2^\resources$ is a function assigning to each agent a
collection of resources---a \emph{bundle}---such that the assigned bundles are
pairwise disjoint. An allocation is \emph{complete} if every resource belongs to some
bundle. 

A \emph{directed graph}~$\graph$ consists of a set~$V$ of \emph{vertices} and
a set~$E \subseteq V \times V$ of \emph{arcs} connecting the vertices; we do not
allow self-loops (i.e., there are no 
arcs of form~$(v,v)$ for any vertex~$v \in V$).
A (simple) \emph{undirected graph}~$\graph{}=(V, E)$ consists of a set~$V$ of vertices and a set~$E$ of
distinct size-$2$ subsets of vertices called~\emph{edges}. 
An~\emph{underlying
undirected graph} of a directed graph~\graph{} is the graph obtained by replacing all (directed)
arcs with (undirected) edges. 
We say an undirected graph~$\graph{}=(V, E)$ is a \emph{clique} if $E=\binom{V}{2}$ and a directed graph is a \emph{bidirectional clique} if $E = V \times V$.
For some vertex~$v \in V$, the set~$\incidentTo(v)$
of~\emph{incident} arcs (edges) is the set of all arcs (edges) with an endpoint
in~$v$.

\subsection{Sharing Model}
We fix an initial allocation~\allocation{} of resources in~\resources{} to agents in~\agents{}.
A~\emph{sharing graph} is an undirected graph~\shareGraph{}=(\agents,\shareRelations) with vertices being the agents; it models possible sharings between the agents.
The following definition of sharing says that two agents can only share resources held by one of them.

\begin{definition} \label{def:sharing_allocation}
 Function~$\sharing_\allocation{} \colon \shareRelations \rightarrow
 2^\resources$ is a \emph{sharing} for~\allocation{} if for every two
 agents~$a_i$ and~$a_j$, with~$\{a_i,a_j\} \in
 \shareRelations$, it holds that
 $
  \sharing_\allocation( \{a_i,a_j\} ) \subseteq
  \allocation(a_i) \cup \allocation(a_j).
 $
\end{definition}

An initial allocation~\allocation{} and a corresponding
sharing~$\sharing_\allocation{}$ %
form a~\emph{sharing allocation}.

\begin{definition}
 A~\emph{sharing allocation} \emph{induced by} allocation~\allocation{} and
 sharing~$\sharing_\allocation$ is a
 function~$\sharingAllocation_{\allocation}^{\sharing_\allocation} \colon
 \agents \rightarrow 2^\resources$ where
 $
 \sharingAllocation_{\allocation}^{\sharing_\allocation}(\agent) :=
 \allocation(\agent) \cup
 \bigcup_{e \in \incidentTo(a)} \sharing_\allocation(e).
 $
\end{definition}

Since the initial allocation~$\allocation$ is fixed, for brevity, we use~$\sharing$ and~$\sharingAllocation{}^\sharing$, 
omitting~$\allocation$ whenever it is not ambiguous.
For simplicity, for every agent~$\agent \in \agents$, we also refer
to~$\sharingAllocation^\sharing(\agent)$ as a
\emph{bundle} of~\agent{}. 

Naturally, each allocation is also a sharing allocation with a trivial
``empty sharing.''
Observe a subtle difference in the intuitive
meaning of a bundle of an agent between sharing allocations and (simple)
allocations. For sharing allocations, a bundle of an agent represents the
resources the agent has access to and can utilize, not only those that the agent
possesses (as for simple allocations).

\subsection{$2$-sharing}

Definition~\ref{def:sharing_allocation} is very general and only requires that two agents
share resources that one of them already has. 
In particular, Definition~\ref{def:sharing_allocation} 
allows one agent to share the same resource with many other agents; and does not constrain the number of sharings an agent could participate in. 
In this paper, we assume that 
each resource can only be shared by \emph{two agents} and each agent can participate
in at most a bounded number of sharings.
We formally express this requirement in~Definition~\ref{def:2-simple-sharing}.

\begin{definition}
A~\emph{$2$-sharing}~$\sharing{}$ is a sharing where, for any three agents
$a_i$, $a_j$, and $a_k$, it holds that
\[
\sharingAllocation^\sharing(\agent_i) \cap \sharingAllocation^\sharing(\agent_j) \cap 
\sharingAllocation^\sharing(\agent_k)=\emptyset.
\]
A \emph{$b$-bounded} $2$-sharing~$\sharing{}$ is a $2$-sharing where, for each agent
$a$, it holds that
 $
 \left|\bigcup_{e \in \incidentTo(a)} \sharing(e)\right| \leq b.
 $
 A~\emph{simple $2$-sharing}~$\sharing{}$ is a $1$-bounded $2$-sharing, i.e., for
 each agent~\agent{}, it holds that
 $
 \left|\bigcup_{e \in \incidentTo(a)} \sharing(e)\right| \leq 1.
 $
 \label{def:2-simple-sharing}
\end{definition}

Herein, we count the number of sharings an agent participate in by the number of resources shared with other agents (either shared to other agents or received from other agents).
Notably, 
in simple $2$-sharing, each agent can \emph{either share} or \emph{receive} a \emph{single resource}.
Thus, every simple $2$-sharing can be interpreted as matching in which each edge is labeled with a shared item.

\subsection{Welfare and Fairness Measures}
We assume agents having additive utility functions.
For an agent~\agent{} with \emph{utility function}~$\utilityFunction \colon
\resources \rightarrow \naturals_{0}$ and a bundle~$R \subseteq \resources$, let
$\utilityFunction(R) := \sum_{\resource \in R} \utilityFunction(\resource)$ be
the~\emph{value} of~$R$ as perceived by~\agent{}. Let us fix a sharing
allocation~$\sharingAllocation^\sharing$ of resources~\resources{} to
agents~\namedorderedsetof{\agents}{\agent}{\agentsCount} with corresponding
utility functions~\orderedlistingof{\utilityFunction}{\agentsCount}.
The~\emph{utilitarian social welfare} of~$\sharingAllocation^\sharing$
is
$$\utilitarianWelfare(\sharingAllocation^\sharing) := \sum_{i \in [\agentsCount]}
 \utilityFunction_{i}(\sharingAllocation^\sharing(\agent_i)).$$ 
The \emph{egalitarian
social welfare} of~$\sharingAllocation^\sharing$ is
$$\egalitarianWelfare(\allocation) :=
\min_{i \in [\agentsCount]} \utilityFunction_{i}(\sharingAllocation^\sharing(\agent_i)).$$
Notice that we assume each agent $a_i$ gets the full utility for all resources in $\sharingAllocation^\sharing(\agent_i)$.
We will discuss a generalization of this assumption in Section \ref{sec:extension}.

A directed graph~\attentionGraph{}=(\agents,\attentionRelations) with vertices being the agents is
an~\emph{attention graph}; it models social relations between the agents.
We say that an agent~$a_i$ \emph{looks at} another agent~$a_j$ if~$(a_i,a_j) \in \attentionRelations$.
An agent is \emph{envious} on \attentionGraph{} under~$\sharingAllocation^\sharing$ if it prefers a bundle of another agent it looks at over its own
one; formally, $\agent_i$ \emph{envies}~$\agent_j$ if~$\utilityFunction_i(
\sharingAllocation^\sharing(\agent_i) ) < \utilityFunction_i(
\sharingAllocation^\sharing(\agent_j) )$ and~$(a_i,a_j) \in \attentionRelations$. We denote the set of envious agents
in~$\sharingAllocation^\sharing$ as~$\Env(\sharingAllocation^\sharing)$. 
For a given (directed) attention
graph~\attentionGraph{} over the agents, a sharing allocation
is~\emph{\attentionGraph-envy-free} if no agent envies its out-neighbors.

\subsection{Useful Problems}
We define two
variants of \textsc{Maximum Weighted Matching} and show that they are solvable
in polynomial time. We will later use them to show polynomial-time solvability
of our problems. Given an edge-weighted undirected graph~$G=(V,E)$ with weight
function~$w\colon E \rightarrow \mathbb{N}$ and two integers~$k_1,k_2 \in \mathbb{N}$,
\SBMWM{} (\SBMWMShort) asks
whether there is a matching~$M$ such
that~$|M|\le k_1$ and~$w(M) \ge k_2$, and \WBMM{} (\WBMMShort) asks 
whether there is a matching~$M$ such that~$w(M)\le k_1$ and~$|M| \ge k_2$,
where~$w(M)=\sum_{e \in M}w(e)$.

\begin{lemma}
\label{lem:matching_p}
\SBMWMShort{} and \WBMMShort{} are solvable in polynomial time.
\end{lemma}
\begin{proof}
For \SBMWMShort{}, we show that it can be reduced to \textsc{Maximum Weighted Matching}, which is known to be solvable in polynomial time.
For any instance of \SBMWMShort{}, since~$w(e) \ge 0$ for all $e \in E$, there is a matching~$M$ such
that~$|M|\le k_1$ and~$w(M) \ge k_2$ if and only if there is a matching~$M$ such
that~$|M|= k_1$ and~$w(M) \ge k_2$.
We add $n-2k_1$ new vertices into the graph and connect them with all the old vertices by edges of weight $C=\sum_{e \in E}w(e)+1$.
Notice that the weight of every new edge is larger than that of every old edge.
Then it is easy to see that every maximum weighted matching in the new graph consists of $k_1$ old edges and $n-2k_1$ new edges, and the $k_1$ old edges form a matching with the maximum weight among all matchings of size $k_1$ in the old graph.
Thus we can apply the algorithm for \textsc{Maximum Weighted Matching} on the new graph and then check whether the sum of the weights of these $k_1$ edges in the old graph is at least $k_2$.
Therefore, \SBMWMShort{} is solvable in polynomial time.

For \WBMMShort{}, we first convert the weight function as follows.
Let~$W=\max_{e \in E}w(e)$ and build a new weight function~$w':E \rightarrow \mathbb{N}$ where~$w'(e)=W-w(e)$ for each~$e \in E$.
Note that~$w(e) \ge 0$ and~$w'(e) \ge 0$ for each~$e \in E$.
Since~$w(e) \ge 0$, there is a matching~$M$ with~$w(M) \le k_1$ and~$|M|\ge k_2$ if and only if there is a matching~$M$ with~$w(M) \le k_1$ and~$|M|= k_2$.
When~$|M|= k_2$, we have $w(M)=Wk_2-w'(M)$, and hence~$w(M) \le k_1$ if and only if~$w'(M) \ge Wk_2-k_1$.
Now the problem is to decide whether there is a matching~$M$ with~$|M|=k_2$ and~$w'(M)\ge Wk_2-k_1$.
Since $w'(e) \ge 0$, this is equivalent to decide whether there is a matching~$M$ with~$|M| \le k_2$ and~$w'(M)\ge Wk_2-k_1$, which is an instance of \SBMWMShort{} and, as shown in the last paragraph, is solvable in polynomial time.
\end{proof}

\section{Improving Social Welfare by Sharing}

In this section, we study the problem of improving utilitarian (and egalitarian)  welfare through sharing, defined as follows.

\begin{definition}\label{def:Pwelfare}
Given an initial complete allocation~\allocation{} of resources~\resources{} to agents~\agents{}, a sharing graph~\shareGraph{}, 
and a non-negative integer~$k$,
$b$-Bounded Utilitarian Welfare Sharing Allocation ($b$-UWSA) asks if there is a $b$-bounded $2$-sharing~$\sharing$  such that $\utilitarianWelfare(\sharingAllocation^\sharing) \geq k$;
$b$-Bounded Egalitarian Welfare Sharing Allocation ($b$-EWSA) asks if there is a $b$-bounded $2$-sharing~$\sharing$  such that $\egalitarianWelfare(\sharingAllocation^\sharing) \geq k$.
\end{definition}

We first consider $b$-\utilProblemShort{}.
When $b=1$, since every simple $2$-sharing corresponds to a matching, we can easily reduce $1$-\utilProblemShort{} to \textsc{Maximum Weighted Matching}.
Thus, $1$-\utilProblemShort{} is solvable in polynomial time.
When $b > 1$, however, the problem is not just finding $c$ matchings such that the total is maximized.
Notice that in $2$-sharing each resource can only be shared once.
Nevertheless, we show that we can still reduce $b$-\utilProblemShort{} to \textsc{Maximum Weighted Matching} via a more involved reduction.

\begin{theorem}
\label{thm:UWSA-P}
$b$-\utilProblemShort{} is solvable in polynomial time for any $b \ge 1$.
\end{theorem}

\begin{proof}
Given an instance of $b$-\utilProblemShort{}, we construct an instance $(G=(V,E),w)$ of \textsc{Maximum Weighted Matching} as follows.
For simplicity, assume each agent $a_i \in \agents$ has at least $b$ resources in the initial allocation $\pi$; otherwise we can ensure this by adding enough resources that are valued as 0 by all agents.
For each agent $a_i \in \agents$ and each resource $r_j \in \resources$, we add a vertex $v_i^j$ into $V$.
In addition, for each agent $a_i \in \agents$, we add $n_i=|\pi(a_i)|-b$ dummy vertices $\{\overline{v}_i^{k}\}_{k=1,2,\dots,n_i}$ into $V$.
For each pair of agents $a_{i_1},a_{i_2}$, for each vertex $v_{i_1}^{j_1}$ corresponding to $a_{i_1}$ and each vertex $v_{i_2}^{j_2}$ corresponding to $a_{i_2}$, we add an edge between $v_{i_1}^{j_1}$ and $v_{i_2}^{j_2}$ with weight $\max\{u_{i_1}(r_{j_2}),u_{i_2}(r_{j_1})\}$.
Finally, for each $v_i^j$ and each $\overline{v}_i^{k}$ corresponding to the same agent $a_i$, we add a dummy edge between them with weight $W=\max_{a_i \in \agents,r_j \in \resources}u_i(r_j)$.
This finishes the construction of the instance $(G=(V,E),w)$.
Notice that there always exists a maximum weighted matching that contains $n_i$ dummy edges for each agent $a_i$.
Let $P=W\sum_{a_i \in \agents}n_i$
 be the weight of those edges.
Next we show that there is a $b$-bounded $2$-sharing $\sharing{}$ such that $\utilitarianWelfare(\sharingAllocation^\sharing) \geq k$ if and only if there is matching~$M$ in graph~$G$ with weight~$\sum_{e \in M}w(e) \ge k-\utilitarianWelfare(\pi)+P$, where $\utilitarianWelfare(\pi)=\sum_{a_i \in \agents}u_i(\pi(a_i))$.

$\Rightarrow:$
Assume there is a $b$-bounded $2$-sharing $\sharing{}$ such that $\utilitarianWelfare(\sharingAllocation^\sharing) \geq k$.
Based on $\sharing{}$, we can find a matching $M$ with the claimed weight by including the corresponding edges for each $e \in \shareRelations$ with $\sharing(e) \neq \emptyset$ and edges between the remaining normal vertices and all dummy vertices.
Formally, for each edge $(a_{i_1},a_{i_2}) \in \shareRelations$ such that $\sharing(a_{i_1},a_{i_2})=r_{j_1}$ and $r_{j_1} \in \pi(a_{i_1})$ we add the edge $(v_{i_1}^{j_1},v_{i_2}^{j_2})$ into matching $M$, where $r_{j_2} \in \pi(a_{i_2})$ is an arbitrary resource except for the resource that is shared by $a_{i_2}$ with some other agent under $\sharing{}$.
Notice that $w(v_{i_1}^{j_1},v_{i_2}^{j_2})=\max\{u_{i_1}(r_{j_2}),u_{i_2}(r_{j_1})\}\ge u_{i_2}(r_{j_1})$.
After this, for each agent $a_i$ there are at least $n_i$ normal vertices not matched and exactly $n_i$ dummy vertices not matched, so we can add $n_i$ dummy edges of the same weight $W$ into $M$, each containing one normal vertex and one dummy vertex.
Since $\utilitarianWelfare(\sharingAllocation^\sharing) \geq k$, we have that 
$\sum_{e \in M}w(e)= \utilitarianWelfare(\sharingAllocation^\sharing)-\utilitarianWelfare(\pi)+P \ge k-\utilitarianWelfare(\pi)+P$.

$\Leftarrow:$
Assume there is matching~$M$ in graph~$G$ with the claimed weight.
Without loss of generality, we can assume $M$ contains $n_i$ dummy edges for every agent $a_i$ since the weight of dummy edges is no smaller than that of non-dummy edges.
Then for each agent $a_i$, $M$ contains at most $c$ non-dummy edges.
Based on these non-dummy edges in $M$ we can find the corresponding $b$-bounded $2$-sharing $\sharing{}$ such that $\utilitarianWelfare(\sharingAllocation^\sharing) \geq k$ as follows.
For each non-dummy edge $(v_{i_1}^{j_1},v_{i_2}^{j_2}) \in M$, we set $\sharing(a_{i_1},a_{i_2})=r_{j_1}$ if $u_{i_2}(r_{j_1}) \ge u_{i_1}(r_{j_2})$ and $\sharing(a_{i_1},a_{i_2})=r_{j_2}$ otherwise.
Since $M$ contains at most $b$ non-dummy edges for each $a_i$, $a_i$ participates in at most $b$ sharings in $\sharing{}$.
Moreover, we have that 
\[
\begin{aligned}
\utilitarianWelfare(\sharingAllocation^\sharing) 
&=\utilitarianWelfare(\pi) + \sum_{e=(v_{i_1}^{j_1},v_{i_2}^{j_2}):\sharing(a_{i_1},a_{i_2}) \neq \emptyset} w(e) \\
&=\utilitarianWelfare(\pi) + \sum_{e \in M}w(e) - P \ge k.
\end{aligned}
\]
\end{proof}

Next we consider $b$-\egalProblemShort{}.
When $b=1$, we show in Lemma \ref{lem:EWSA-1-P} that we can reduce the problem to \textsc{Maximum Matching}.
The idea is to partition all agents into two subgroups according to the target $k$ 
and build a bipartite graph characterizing whether one agent from one group can improve 
the utility of one agent from the other group to $k$ by sharing, then the problem is just to 
find a maximum matching on the bipartite graph.

\begin{lemma}
\label{lem:EWSA-1-P}
$1$-\egalProblemShort{} is solvable in polynomial time.
\end{lemma}
\begin{proof}
 \newcommand{\utilityThreshold}{\ensuremath{k}}
 \newcommand{\betterAgents}{\ensuremath{\agents^+_\utilityThreshold}}
 \newcommand{\worseAgents}{\ensuremath{\agents^-_\utilityThreshold}}
  Depending on the target~\utilityThreshold{}, we partition the set~\agents{} of agents into two
 sets \betterAgents{} and~\worseAgents{} containing, respectively, the agents
 with their bundle value under~$\allocation$ at least~\utilityThreshold{} and smaller than~\utilityThreshold{}.
 Now, we construct a bipartite, undirected graph~$\graph_\utilityThreshold =
 (\betterAgents, \worseAgents, E_\utilityThreshold)$. Consider two
 agents~$\agent_i \in \betterAgents$ and~$\agent_j \in \worseAgents$ that are
 neighbors in the sharing graph~\shareGraph{}. An edge~$e = \{\agent_i, \agent_j\}$ belongs
 to~$E_\utilityThreshold$ if~$\agent_i$ can share a resource with~$\agent_j$ to
 raise the utility of the latter to at least~\utilityThreshold{}; 
 formally,
 there exists a resource~$\resource \in \allocation(\agent_i)$ such
 that~$\utilityFunction_j(\allocation(\agent_j)) + \utilityFunction_j(\resource)) \geq \utilityThreshold$.

 We claim that there is a simple $2$-sharing~$\delta$ with~$\egalitarianWelfare(\sharingAllocation^\sharing) \geq k$ if and only if there is matching~$M$ in
 graph~$G_k$ with~$|M| \ge |\worseAgents|$.
 The backward direction is clear according to the construction of~$\graph_\utilityThreshold$.
 For the forward direction, if there is a simple $2$-sharing~$\delta$ with~$\egalitarianWelfare(\sharingAllocation^\sharing) \geq k$, we can build a matching $M$ in
 graph~$G_k$ by adding an edge $e = \{\agent_i, \agent_j\}$ to $M$ if $\delta(\{\agent_i,\agent_j\}) \neq \emptyset$.
 Notice that since $\egalitarianWelfare(\sharingAllocation^\sharing) \geq k$, the edge~$e = \{\agent_i, \agent_j\}$ must be in~$E_\utilityThreshold$.
 Since~$\egalitarianWelfare(\sharingAllocation^\sharing) \geq k$, according to construction of~$\graph_\utilityThreshold$ and~$M$, we have that for each~$\agent_j \in \worseAgents$, there exists an agent~$\agent_i \in \betterAgents$ such that $\{\agent_i, \agent_j\} \in M$, and hence $|M| \ge |\worseAgents|$.
 Thus, we just need to check whether the maximum matching in 
 graph~$G_k$ has size at least~$|\worseAgents|$, which is solvable in polynomial time.
\end{proof}

On the other hand, we show that $b$-\egalProblemShort{} is NP-hard when $b \ge 2$.
Notice that there is a trivial reduction from \textsc{3-Partition} to $b$-\egalProblemShort{} if $b \ge 3$.
To show the result for any $b \ge 2$, we reduce from the strongly NP-hard \textsc{Numerical Three-dimensional Matching} (N3DM) \cite{DBLP:books/fm/GareyJ79}.

\begin{theorem}
\label{thm:EWSA-b-hard}
$b$-\egalProblemShort{} is NP-hard for any constant $b \ge 2$.
\end{theorem}

\begin{proof}
We present a polynomial-time many-one reduction from the \nphard{} N3DM.
Therein, given 3 multisets of positive integers~$X,Y,Z$, each containing $m$ elements, and a bound $T$, the task is to decide whether there is a partition $S_1,S_2,\dots,S_m$ of $X \cup Y \cup Z$ such that each $S_i$ contains exactly one element from each of $X,Y,Z$ and the sum of numbers in each $S_i$ is equal to $T$.
Given an instance $(X,Y,Z)$ of N3DM, we construct an instance of $b$-\egalProblemShort{} as follows.
Without loss of generality, assume all elements from $X \cup Y \cup Z$ are smaller than $T$ and the sum of them is equal to $B \coloneqq mT$. 
We set the goal $k=(B^2+B+1)T$.
We create 3 groups of agents corresponding to the 3 multisets $X,Y,Z$.
For each $x_i \in X$, we create an agent $a_i^1$ in group 1 who holds a \emph{large} resource that is valued as $k$ by itself and $B^2T+x_i$ by all other agents.
For each $y_i \in Y$, we create an agent $a_i^2$ in group 2 who holds a \emph{middle} resource that is valued as $k$ by itself and $BT+y_i$ by all other agents.
For each $z_i \in Z$, we create an agent $a_i^3$ in group 3 who holds a \emph{small} resource that is valued as $z_i$ by itself and 0 by all other agents.
In the initial allocation, all $2m$ agents in group 1 and group 2 have utility exactly $k$ and all $m$ agents from group 3 have utility less than $k$.

$\Rightarrow:$ 
If $(X,Y,Z)$ is a ``yes''-instance of N3DM, then based on the partition $S_1,S_2,\dots,S_m$ with the claimed properties, we can find a sharing~$\sharing{}$ such that each agent participates in at most two sharings and has utility at least $k$ under $\sharingAllocation^\sharing$ as follows.
For each $S_i=\{x_{i_1},y_{i_2},z_{i_3}\}$ with $x_{i_1}+y_{i_2}+z_{i_3}=T$, we let agent $a_{i_1}^1$ share its big resource with $a_{i_3}^3$ and let agent $a_{i_2}^2$ share its middle resource with $a_{i_3}^3$ such that agent $a_{i_3}^3$ has utility $B^2T+x_{i_1}+BT+y_{i_2}+z_{i_3}=(B^2+B+1)T=k$.
Since each $S_i$ contains exactly one element from each of $X,Y,Z$,
we have that each agent in group 3 has value exactly $k$ under $\sharingAllocation^\sharing$.
Thus all agents have value at least $k$ under $\sharingAllocation^\sharing$.
Moreover, every agent in group 1 and group 2 participates in exactly one sharing and every agent in group 3 participates in exactly two sharings.

$\Leftarrow:$
If the instance of $b$-\egalProblemShort{} is a ``yes''-instance, then there is a sharing~$\sharing{}$ such that every agent has utility at least $k$ under $\sharingAllocation^\sharing$.
For each agent $a_i^3$ in group 3, let $S'_i=\sharingAllocation^\sharing(a_i^3)$ be the set of resources $a_i^3$ has access to under $\sharingAllocation^\sharing$.
Since each agent $a_i^3$ in group 3 values all small items hold by other agents as 0, without loss of generality, we can assume $S'_i$ contains no small items except for the one initially hold by $a_i^3$ itself.
Notice that from the viewpoints of agents in group 3, the sum of value of all middle resources is smaller than one big resource.
Since agent $a_i^3$ has utility at least $k$ under $\sharingAllocation^\sharing$, $S'_i$ contains at least one big resource.
It follows that each $S'_i$ contains exactly one big resource as each big resource can be shared by at most one agent from group 3 and there are $m$ agents in group 3 and $m$ big resources.
Next, to guarantee that each agent $a_i^3$ has utility at least $k$, each $S'_i$ should contain at least one middle resource.
Again, since each middle resource can be shared by at most one agent and there are $m$ agents in group 3 and $m$ middle resources, each $S'_i$ contains exactly one middle resource.
Thus, each $S'_i$ contains exactly one big resource, one middle resource, and one small resource, and each resource is contained in exactly one $S'_i$.
According to the construction, based on $S'_1,S'_2,\dots,S'_m$, we can find a partition $S_1,S_2,\dots,S_m$ of $X \cup Y \cup Z$ such that each $S_i$ contains exactly one element from each of $X,Y,Z$ and the sum of numbers in each $S_i$ is at least $k-B^2T-BT=T$.
Since the sum of all elements from $X \cup Y \cup Z$ is $mT$, the sum of numbers in each $S_i$ is exactly $T$.
Therefore, $(X,Y,Z)$ is a ``yes''-instance of N3DM.
\end{proof}

\section{Reducing Envy by Sharing}

In this section, we study the problem of reducing envy through sharing, defined as follows.

\begin{definition}\label{def:Penvy}
Given an initial complete allocation~\allocation{} of resources~\resources{} to agents~\agents{}, a sharing graph~\shareGraph{}, 
an attention graph~\attentionGraph{}, 
and a non-negative integer~$k$,
Envy Reducing Sharing Allocation (ERSA) asks if there is a \emph{simple $2$-sharing}~$\sharing$  such that  the number of envious agents~$|\Env(\sharingAllocation^\sharing)| \leq k$.
\end{definition}

Notice that in the above definition we restrict that the sharing is a \emph{simple $2$-sharing}, i.e., $1$-bounded $2$-sharing, because the problem for reducing envy in this setting is already~\nphard{}, even in a special case when the attention graph
and the sharing graph are (bidirectional) cliques
and the goal is to decrease the number of envious agents
by one, as shown in Theorem \ref{thm:hard_clique}.

\begin{theorem}%
 \envyReducingProblemShort{} is~\nphard{} even if 
 the attention graph and the sharing graph are (bidirectional) cliques,
 and the goal is to reduce the number of envious agents by at least one.
 \label{thm:hard_clique}
\end{theorem}
\begin{proof}
 We present a polynomial-time many-one reduction from the NP-hard \IS{} problem.
 Therein, for an undirected
 graph~$G=(V,E)$ and an integer~$\ell$, we ask
 whether there is a subset of~$V$ of at least~$\ell$~vertices
 that are pairwise non-adjacent.
 
 For an instance~$(G=(V,E),\ell)$ of \IS{} with \namedorderedsetof{V}{v}{n}, we
 construct an instance of \envyReducingProblemShort{} as follows. For each
 vertex~$v_i \in V$, we create an agent~$a_i$ who initially has two
 resources~$r_i^a$ and~$r_i^{a'}$. Moreover, we add~$\ell$~ more agents:
 \emph{providers}~$p_1$, $p_2$, $\dots$, $p_{\ell-1}$ and a \emph{special provider}~$s$.
 Initially, we allocate resource~$r_i^p$ to each provider~$p_i$ and
 resource~$r^s$ to the special
 provider. Next we specify utility functions
 (see also Table~\ref{tab:hard_clique}). Each agent~$a_i$ has value~$1$ for each
 resource~$r_j^a$ for which~$\{v_i,v_j\} \in E$, $j \in [n]$, and value~$0$ for
 other resources~$r_j^a$ (note that, in particular, agent~$a_i$ gives value~$0$
 to its own resource).
 Each agent~$a_i$ has value 0 for all~$r_j^{a'}$, $j \in [n]$, and value 3 for
 other resources. Each provider~$p_i$ gives
 value~$1$ to resources in~$\{r_1^a,r_2^a,\dots,r_n^a\} \cup \{r_1^{a'},r_2^{a'},\dots,r_n^{a'}\}$ and value~$3$
 to other resources. The special provider~$s$ has value~$3$ for its own
 resource~$r^s$ and value~$0$ for all other ones. The attention graph is
 a (bidirectional) clique  and the sharing graph is also a clique, so every two agents can share their resources.
 By the construction, initially there are~$n$ envious agents:
 $a_1,a_2,\dots,a_n$.
 Thus, we end the construction
 by setting the target number~$k=n-1$, so
 we aim at
 decreasing the number of envious agents by at least~$1$.
 
 In what follows we show that there is an independent set of size at
 least~$\ell$ in~$G$ if and only if there is simple $2$-sharing~$\delta$ such
 that~$|\Env(\sharingAllocation^\delta)| \leq k$.
 Notice that in every simple $2$-sharing, the special provider~$s$ is not
 envious.

\begin{table}[t]
 \centering
 \begin{tabular}{m{0.5cm}|cccc}
  \toprule
  $u$ &  $r_j^a$ & $r_j^{a'}$ & $r_j^p$ & $r^s$  \\ 
  \midrule
  $a_i$ & $\left[ \{v_i, v_j\} \in E \right]$ & 0 & 3 & 3\\
  $p_i$ & 1 & 1 & 3 & 3\\
  $s$ & 0 & 0 & 0 & 3\\
 \bottomrule
 \end{tabular}
 \caption{Utility functions in the proof of~Theorem~\ref{thm:hard_clique}. We use the
 Iverson bracket notation: for some logical expression~$X$,
 $\left[X\right]$ is one if~$X$ is true and zero otherwise.}
\label{tab:hard_clique}
\end{table}

$\Rightarrow:$ Suppose that there is an independent set~$S$ of size~$\ell$
in~$G$. Let~$S'$ be a set of $\ell$~agents from~\orderedsetof{a}{n} corresponding to~$S$.
Then, each agent in~$S'$ can share with a different provider (including the
special provider) and increase its own value to~$3$. After this sharing, denoted
by~$\delta$, agents in~$S'$ do not envy any provider since from their point of
view each provider has a bundle of value~$3$. In addition, since vertices in~$S$
are pairwise non-adjacent, agents in~$S'$ do not envy each other as they see
each other having a bundle of value $3+0+0=3$. Together with the special provider,
there are at least~$\ell+1$ non-envious agents. Consequently,
$|\Env(\sharingAllocation^\delta)| \leq k$.

$\Leftarrow:$ Suppose that there is a simple $2$-sharing~$\delta$ such
that $|\Env(\sharingAllocation^\delta)| \leq k$.
Let~$N$ be the set of non-envious agents excluding the special provider~$s$ after the sharing.
Since overall there are~$n+\ell$ agents and at most~$k$ of them are envious, we have~$|N| \ge (n+\ell)-k-1=\ell$ ($s$ is not included in~$N$).
Denote~$N_a=N \cap \orderedsetof{a}{n}$ and~$N_p=N \cap \orderedsetof{p}{\ell-1}$. 
Then~$|N_a|+|N_p|=|N| \ge \ell$.
Since~$|N_p| \le \ell-1$, we have~$N_a \neq \emptyset$.
We will show that~$N_p=\emptyset$.
Suppose, towards a contradiction, that~$N_p \neq \emptyset$ and let~$a_{i^*} \in N_a$ and~$p_{j^*} \in N_p$.
Initially, $u_{a_{i^*}}(\pi(a_{i^*}))=0$ and~$u_{a_{i^*}}(\pi(p_{j^*}))=3$.
Since~$a_{i^*}$ is not envious after the sharing~$\delta$, $a_{i^*}$ must share with a provider.
Consequently, $u_{p_{j^*}}(\sharingAllocation^\delta(a_{i^*}))=1+1+3=5$.
Since initially~$u_{p_{j^*}}(\pi(p_{j^*}))=3$ and~$p_{j^*}$ is not envious after the sharing~$\delta$, $p_{j^*}$ must share with another provider.
However, this will make~$a_{i^*}$ envious again.
Thus, $N_p=\emptyset$, and hence~$N=N_a  \subseteq \{a_1,a_2,\dots,a_n\}$.
Now to make all agents in~$N$ non-envious, all of them have to share with one provider.
For any two agents~$a_i,a_j \in N$, let~$r_{i'}^p$ and~$r_{j'}^p$ be the resources shared to~$a_i$ and~$a_j$, respectively.
Since~$a_i$ is not envious towards~$a_j$, we have that~$u_{a_i}(\sharingAllocation^\delta(a_i)) \ge u_{a_i}(\sharingAllocation^\delta(a_j))$, where~$\sharingAllocation^\delta(a_i)=\{r_i^a,r_i^{a'},r_{i'}^p\}$ and~$\sharingAllocation^\delta(a_j)=\{r_j^a,r_j^{a'},r_{j'}^p\}$.
Since~$u_{a_i}(r_i^a)=u_{a_i}(r_i^{a'})=u_{a_i}(r_j^{a'})=0$ and~$u_{a_i}(r_{i'}^p)=u_{a_i}(r_{j'}^p)=3$, we
get~$u_{a_i}(r_j^a)=0$, which means $\{v_i, v_j\} \not \in E$. Thus the
corresponding vertices for agents in~$N$ form an independent set of size at
least~$\ell$ in~$G$.
\end{proof}

Theorem~\ref{thm:hard_clique} in fact constitutes a strong intractability result
and it calls for further studies on other specific features of the
input. %
We counteract the intractability result of
\envyReducingProblemShort{} (Theorem~\ref{thm:hard_clique}) by considering cases with
few agents, tree-like graphs, identical utility functions, or few resources.

\subsection{Reducing Envy for Few Agents}

The simple fact that for $n$~agents and $m$~resources there are at most $m^n$
possible $2$-sharings leads to a straightforward brute-force algorithm that runs in
polynomial time if the number of agents is constant.
\begin{observation}
\envyReducingProblemShort{} with $n$~agents and $m$~resources is solvable
in~$O(\binom{m}{\lfloor \frac{n}{2} \rfloor}^nmn)$ time.
\label{obs:XP_agents}
\end{observation}

\begin{proof}
 We can simply enumerate all possible sharings and compute for each sharing the
 number of envious agents. A sharing contains at most~$\lfloor \frac{n}{2}
 \rfloor$ shared resources and each shared resource has at most~$n$
 destinations, thus there are at most~$\binom{m}{\lfloor \frac{n}{2} \rfloor}^n$
 different sharings. Computing the number of envious agents under a sharing can
 be done in~$O(mn)$ time.
\end{proof}

However, due to the factor~$m$ in the base, the running time of such an
algorithm skyrockets 
with large number of resources (even for small, fixed
values of~$n$).
We improve this by showing that
\envyReducingProblemShort{} is fixed-parameter tractable with respect to the
number of agents.

\begin{theorem} \label{thm:fpt-agents}
 \envyReducingProblemShort{} with $n$~agents and $m$~resources is solvable
 in~$O((2n)^{n}m^2)$ time.
\end{theorem}

The high-level idea behind Theorem~\ref{thm:fpt-agents} is as follows. In order to find a
desired sharing, our algorithm guesses \emph{target agents}---a set of at least
$n-k$~agents that do not envy in the desired sharing---and a \emph{sharing
configuration}---a set of ordered pairs of agents that share some resource in
the desired sharing. Then, for such a guessed pair, the algorithm 
tests whether the desired sharing indeed exists. 
If it is true for at least one guessed pair,
then the algorithm returns~``yes''; otherwise, it returns~``no.''
The main difficulty in checking the existence of the desired sharing is 
that we need to maintain the envy-freeness within target agents while 
increasing their utilities.

Before stating
the algorithm more formally, we give some notation and definitions. Let~$C$ be a
fixed set of target agents.

\begin{definition}
 A \emph{sharing configuration}~$M$ for a set~$C$ of target agents is a set of
 arcs such that
\begin{enumerate}
\item $M$ is a set of vertex-disjoint arcs and
\item if $(i,j) \in M$, then $\{i,j\} \in \shareRelations$ and $j \in C$.
\end{enumerate}
A simple $2$-sharing~$\delta$ is called a \emph{realization} of~$M$ if~$\delta$ only specifies 
the shared resource for each arc in~$M$; formally, for each~$(i,j) \in M$, we have that
$\big(\delta(\{i,j\}) \neq \emptyset\big) \land \big( \delta(\{i,j\}) \in \pi(i) \big)$,
and for each~$\{i,j\}$ with $\delta(\{i,j\}) \neq \emptyset$, we have that 
$\big((i,j) \in M \land
\delta(\{i,j\}) \in \pi(i)\big) \lor \big((j,i) \in M \land \delta(\{i,j\}) \in
\pi(j)\big)$.
A realization~$\delta$ is \emph{feasible} if
~$C \cap \Env(\Pi^\delta) = \emptyset$.
 i.e., no agent in $C$ will be envious under $\delta$.
\end{definition}
Note that a sharing configuration does not only describe shares of a proper
simple~$2$-sharing but also ensures that the shared resources are indeed shared ``to''
the target agents; we justify this restriction later in Lemma \ref{lem:configuration-is-enough}.

Let us fix a sharing configuration~$M$ 
for~$C$. For each target agent~$a_i$, we
define a set~$P^0_i$ of \emph{initially possible resources} that~$a_i$ might get
in some realization of~$M$. For convenience, we augment each~$P^0_i$ with a dummy
resource~$d_i$ that has utility zero for every agent. Formally, we have
\begin{equation}
\label{eq:S_i}
  P^0_i :=
    \begin{cases*}
      \pi(j) \cup \{d_i\} & if $\exists j$ such that $(j,i) \in M$,\\
      \{d_i\}        & otherwise.
    \end{cases*}
 \end{equation}

 For each target agent~$a_i \in C$, we define a \emph{utility
 threshold}~$t_i$---the smallest utility agent~$a_i$ must have such that~$a_i$ will not envy agents
 \emph{outside}~$C$.
Formally, if there is at least one agent $a_j \not \in C$ such that $(a_i, a_j) \in  \attentionGraph$, then
\begin{equation}
\label{eq:t_i}
t_i \coloneqq \max_{a_j \not \in C,(a_i, a_j) \in  \attentionGraph} u_i(\pi(j)),
\end{equation}
otherwise, $t_i \coloneqq 0$.
 If some target agent cannot achieve its utility threshold by obtaining at most
 one of its initially possible resources, then there is no realization of~$M$
 such that the agent does not envy. We express it more formally
 in~Observation~\ref{obs:threshold-correct}.

 \begin{observation}
 \label{obs:threshold-correct}
 There is no feasible realization of~$M$ in which some agent~$a_i \in C$ gets a
 resource~$r \in P^0_i$ such that $u_i(\pi(i) \cup \{r\}) < t_i$.
 \end{observation}

 For each target agent~$a_i \in C$, we define a set of~\emph{forbidden resources}. 

 \begin{definition} \label{def:blocking_resource}
  Let~$C = \{a_1, a_2, \ldots, a_q\}$ and let $\mathcal{P}=\{P_1, P_2,
  \ldots, P_q\}$ be a family of sets of possible resources 
     for the target agents.
  Then resource~$r \in P_i$ is a~\emph{forbidden resource} for some target
  agent~$a_i$ if there is some target agent~$a_j$ with~$(a_j, a_i) \in
  \attentionGraph$ such that $$\max\{u_j(\pi(j) \cup \{r'\}) \mid r' \in P_j\}
  < u_j(\pi(i) \cup \{r\}),$$
  that is, if agent $a_i$ gets resource $r$, then agent $a_j$ will envy $a_i$ 
  even if $a_j$ gets its most valuable resource from $P_j$.
   We denote the set of all forbidden resources
  for~$a_i$ as~$F_i(\mathcal{P})$.
 \end{definition}

 Observe that in every feasible realization no target agent gets one of its
 forbidden resources since otherwise there is another target agent that
 envies.
 \begin{observation}\label{obs:no_forbidden}
  Let $\mathcal{P}$ be a family of possible resources for the target agents. In
  every feasible realization no target agent~$a_i$ gets a resource
  from~$F_i(\mathcal{P})$.
 \end{observation}

 Based on the above observations, Algorithm~\ref{alg:fpt-agents} tests
 whether for a pair of a set~$C$ of target agents and a sharing
 configuration~$M$ there is a feasible realization. The algorithm keeps track of
 the possible resources~$P_i$ for each target agent~$a_i$. Starting with
 each~$P_i$ equal to the corresponding set of initially possible resources, it
 utilizes~Observation~\ref{obs:threshold-correct} and removes the ``low-utility''
 resources.
 Then, utilizing~Observation~\ref{obs:no_forbidden}, the algorithm finds all forbidden resources for a particular collection of
 the possible resources for the target agents and eliminates the forbidden
 resources. This procedure is repeated exhaustively. Finally, if at least one of
 the possible resource sets is empty, the algorithm outputs ``no.'' Otherwise,
 the algorithm returns ``yes'' since at least one resource remained in the set
 of possible resources for every target agent.

\begin{algorithm}[t]
 \caption{Testing existence of a feasible realization of
  sharing configuration~$M$ for set~$C$~of target agents.
  \label{alg:fpt-agents}}
\textsl{DoesFeasibleRealizationExist}$(\attentionGraph,\pi,\{u_i\}_{i \in C},C,M)$\\
\For{each agent~$a_i \in C$}{
 $P_i \leftarrow P^0_i \setminus \{r \in P_i \mid
 u_i(\pi(i) \cup \{r\})< t_i\}$\;
}
\Repeat{$B = \emptyset$}{
 $B \leftarrow \bigcup_{a_i \in C} F_i(\{P_1, P_2,
 \ldots, P_{|C|}\})$\;
 $P_i \leftarrow P_i \setminus B$\;
}
 \lIf{$\exists i$ with $P_i=\emptyset$}{\Return{``no''} \textbf{else return} ``yes''}
\end{algorithm}

After applying~Observation~\ref{obs:threshold-correct} and \ref{obs:no_forbidden}
 to~Algorithm~\ref{alg:fpt-agents} and proving its
correctness (Lemma~\ref{lem:blocking-correct} and Lemma~\ref{lem:configuration-is-enough}), we finish the proof
of~Theorem~\ref{thm:fpt-agents}.

\begin{lemma}%
\label{lem:blocking-correct}
 In~Algorithm~\ref{alg:fpt-agents}, if some~$P_i$ is empty after the repeat-loop, then
 there is no feasible realization for~$M$ for~$C$.
\end{lemma}

\begin{proof}
Suppose towards a contradiction that $P_i$ is empty after the repeat-loop and there is a feasible realization~$\delta$ for~$M$ such that~$C \cap \Env(\Pi^{\delta})=\emptyset$.
Let~$t^* \in \Pi^{\delta}(i) \setminus \pi(i)$ be the resource shared to agent~$i$.
For any agent~$j \in C$ with~$(j,i) \in \attentionRelations$, let~$t_j\in \Pi^{\delta}(j) \setminus \pi(j)$ be the resource shared to agent~$j$.
We remark that here resources~$t^*$ and~$t_j$ could be dummy resources.
Since~$j \in C$ and~$C \cap \Env(\Pi^{\delta})=\emptyset$, we have~$j \not \in \Env(\Pi^{\delta})$, and hence~$j$ does not envy $i$ after~$\delta$.
Thus~$u_j(\pi(j))+u_j(t_j) \ge u_j(\pi(i))+u_j(t^*)$.
Since this holds for any agent~$j \in C$ with~$(j,i) \in \attentionRelations$, resource~$t^* \in S_i$ is not a forbidden resource, which contradicts
that~$S_i$ becomes empty after deleting blocking resources.
\end{proof}

\begin{lemma}%
\label{lem:configuration-is-enough}
If there is a simple $2$-sharing~$\sigma$ such that
~$C \cap \Env(\Pi^{\sigma})=\emptyset$,
then there is a sharing configuration~$M$ for~$C$ that has a feasible
realization and~Algorithm~\ref{alg:fpt-agents} outputs~``yes''; otherwise, the algorithm
outputs~``no''.
\end{lemma}

\begin{proof}
We first show that if there is a simple $2$-sharing~$\sigma$ such that
~$C \cap \Env(\Pi^{\sigma})=\emptyset$,
then there is a sharing configuration~$M$ for~$C$ that has a feasible realization~$\delta$.
Recall that all sharing configurations have the restriction that the shared resources are indeed shared ``to'' the target agents.
We now justify this restriction.
Let~$E=\{\{i,j\} \in \shareRelations \mid \sigma(\{i,j\}) \neq \emptyset\}$ be the set of edges where there is a sharing in~$\sigma$.
Let $E_1 \subseteq E$ be the set of edges that could be used in a sharing configuration for~$C$, i.e., 
\begin{align*}
E_1=& \{\{i,j\} \in E \mid 
\big(  
\sigma(\{i,j\}) \in \pi(i) \land
j \in C
\big) \\
&\lor
\big( 
\sigma(\{i,j\}) \in \pi(j)  \land i \in C 
\big)\}.
\end{align*}
We construct a new simple $2$-sharing $\delta$ by restricting $\delta$ on $E_1$:
\begin{equation*}
    \delta(\{i,j\}) =
    \begin{cases*}
    \sigma(\{i,j\})        & if $\{i,j\} \in E_1$,\\
      \emptyset & otherwise.
      
    \end{cases*}
 \end{equation*}
Then we define the sharing configuration as
$$M=\{(i,j) \mid \delta(\{i,j\}) \neq \emptyset \land \delta(\{i,j\}) \in \pi(i)\}.$$
Now it is clear that~$M$ is a sharing configuration for~$C$ and~$\delta$ is a realization of~$M$.
Moreover, since all sharings through edges in~$E\setminus E_1$ only increase the bundles of agents in~$\agents \setminus C$, we have that
\begin{equation*}
    \begin{cases*}
      u_i(\Pi^{\delta})=u_i(\Pi^{\sigma}) & $\forall i \in C$,\\
      u_i(\Pi^{\delta}) \le u_i(\Pi^{\sigma})        & otherwise.
    \end{cases*}
 \end{equation*}
Since~$C \cap \Env(\Pi^{\sigma})=\emptyset$, we have that~$C \cap \Env(\Pi^{\delta})=\emptyset$.
Therefore, $\delta$ is a feasible realization.

Next, we show that there is a sharing configuration~$M$ for~$C$ that has a feasible realization if and only if Algorithm~\ref{alg:fpt-agents} outputs~``yes''.
According to Lemma \ref{lem:blocking-correct}, if the algorithm output ``no'', then there is no feasible realization for $M$ for $C$.
On the other hand, if the algorithm output ``yes'', then at the end of the algorithm we get $P_i \neq \emptyset$ for each agent $a_i \in C$.
We can build a simple $2$-sharing $\delta$ by choosing an arbitrary resource from $P_i$ for each agent $a_i \in C$.
It is clear that $\delta$ is a realization for $M$ for $C$.
Moreover, 
for any agent $a_i \in C$, $a_i$ does not envy agents outside $C$ (as resources that do not have high enough utility for $a_i$ such that $a_i$ does not envy agents outside $C$ are removed from $P_i$ and $P_i \neq \emptyset$ at the end) and $a_i$ does not envy agents in $C$ (as $B=\emptyset$).
Therefore, $\delta$ is a feasible realization for $M$ for $C$.
\end{proof}

Eventually, we are set to present the proof of Theorem~\ref{thm:fpt-agents}.

\begin{proof}[Proof of Theorem~\ref{thm:fpt-agents}]
 According to Lemma~\ref{lem:configuration-is-enough}, to solve an instance
 of~\envyReducingProblemShort{}, it is enough to test whether there is a pair of
 a target subset and a sharing configuration that has a feasible realization.
 Since, checking a feasible realization, due
 to~Lemma~\ref{lem:configuration-is-enough}, can be done by Algorithm~\ref{alg:fpt-agents}, we
 check all such possible pairs and return ``yes'' if there is (at least) one
 that has a feasible realization.

 There are $O(2^n)$~possible target sets and
 at most~$n^n$ possible sharing configurations per target set, which
 gives~$O((2n)^n)$ cases. For each case, we apply~Algorithm~\ref{alg:fpt-agents}. Therein, the for-loop takes $O(nm)$~time. Concerning the
 repeat-loop that runs at most $m$~times, computing the set~$B$ takes~$O(nm)$~time;
 thus, the repeat-loop takes
 $O(nm^2)$ time. Finally, Algorithm~\ref{alg:fpt-agents} runs in
 $O(nm^2)$~time and \envyReducingProblemShort{} can be solved in
 $O((2n)^{n}m^2)$~time.
\end{proof}

Next, we show that restricting the parameter~$k$ ``number of envious agents''
does not help to make \envyReducingProblemShort{} solvable in polynomial time.

\begin{theorem}
 \envyReducingProblemShort{} is~\nphard{} even if 
 the goal is to reduce the number of
 envious agents from one to zero.
 \label{thm:decreasing_by_one_nph}
\end{theorem}
\begin{proof}
 \newcommand{\sFormula}{\ensuremath{\phi}}
 \newcommand{\sVariables}{\ensuremath{X}}
 \newcommand{\sVariable}{\ensuremath{x}}
 \newcommand{\sClause}{\ensuremath{C}}
 \newcommand{\sLiteral}{\ensuremath{\ell}}
 \newcommand{\snthLiteralOf}[2]{\ensuremath{\ell_{#2}^{#1}}}
 \newcommand{\sClausesCount}{\ensuremath{\bar{m}}}
 \newcommand{\sVariablesCount}{\ensuremath{\bar{n}}}
 \newcommand{\agentGen}{\ensuremath{a}}
 \newcommand{\leaderOf}[1]{\ensuremath{a_\text{L}({#1})}}
 \newcommand{\followerOf}[1]{\ensuremath{a_\text{F}({#1})}}
 \newcommand{\posAgentOf}[1]{\ensuremath{\agentGen({#1})}}
 \newcommand{\negAgentOf}[1]{\ensuremath{\bar{\agentGen}({#1})}}
 \newcommand{\nthDummyOf}[2]{\ensuremath{d({#2}, {#1})}}
 \newcommand{\donorOf}[1]{\ensuremath{\overrightarrow{d}({#1})}}
 \newcommand{\recipientOf}[1]{\ensuremath{\mathring{d}({#1})}}
 \newcommand{\envyThreshold}{\ensuremath{k}}
 \newcommand{\truthAssignment}{\ensuremath{\mathbf{x}}}

 We provide a polynomial-time many-one reduction from the NP-hard \threeCNFSAT{} problem that
 asks whether a given Boolean expression in conjunctive normal form with each
 clause of size at most three is satisfiable. From now on we stick to the
 following notation. Let~$\sFormula{} = \bigwedge_{i \in [\sClausesCount]}
 \sClause_i$ be a $3$-CNF formula over a
 set~\namedorderedsetof{\sVariables}{\sVariable}{\sVariablesCount} of variables
 where a clause~$\sClause_i$, $i \in [\sClausesCount]$, is of the form
 $(\snthLiteralOf{1}{i} \vee \snthLiteralOf{2}{i} \vee \snthLiteralOf{3}{i})$.
 We use a standard naming scheme and call each~$\snthLiteralOf{\cdot}{i}$ a
 literal.

 We first describe the~\emph{leader gadget}, the~\emph{clause
 gadget}, and the~\emph{variable gadget}.
 See~Figure~\ref{fig:cnf-gadgets} for an overview on the construction.
 Then, we show how to connect our gadgets to
 achieve a desired instance of~\envyReducingProblemShort{}. Eventually, we show the
 reduction's correctness preceded by a discussion on the structure of the built
 instance.
 Notably, in our construction we use a unanimous utility function,
 that is, each agent we introduce has the same utility function over all
 introduced resources.

 \begin{figure}[t]
  \subfloat[The leader gadget for formula~\sFormula{}.]{
   \hspace{2em}%
   \centering%
   \begin{tikzpicture}[
    every node/.style={inner sep=1},
    every edge/.style={draw=black},
    every label/.style={font=\scriptsize,draw=none,label distance=0}]
    \node[label={225:\((2)\)}] (j) at (0,2) {\leaderOf{\sFormula}};
    \node[label={135:\((0)\)}, draw] (i) at (0,0) {\followerOf{\sFormula}};
    \draw[->,arrows={-latex}]  (i) edge (j);
   \end{tikzpicture}%
   \hspace{2em}%
  }%
  \hfill%
  \subfloat[The variable gadget for variable~\sVariable{}.]{
   \centering%
   \begin{tikzpicture}[
    every node/.style={inner sep=1},
    every edge/.style={draw=black,font=\scriptsize},
    every label/.style={font=\scriptsize,draw=none,label distance=0}]
    \newcommand\xend{3}
	   \node[label={[label distance=-3pt]210:\((3)\)}] (dx) at (.5*\xend,2) {\nthDummyOf{1}{\sVariable}};
    \node[label={below:\((1,2)\)}] (vx) at (.5*\xend,1)
    {\nthDummyOf{2}{\sVariable}};
    \node[label={above:\((0)\)}] (xx) at (0,0) {\posAgentOf{\sVariable}};
    \node[label={above:\((0)\)}] (nx) at (\xend,0) {\negAgentOf{\sVariable}};
    
    \draw[->,arrows={-latex}]
      (vx) edge (xx)
      (dx) edge (xx)
      (vx) edge (nx)
      (dx) edge (nx);
   \end{tikzpicture}%
  }%
  \hfill%
  \subfloat[The clause gadget for clause~$\sClause_i = (\snthLiteralOf{1}{i}
   \vee \snthLiteralOf{2}{i} \vee \snthLiteralOf{3}{i})$.]{
   \centering%
   \begin{tikzpicture}[
    every node/.style={inner sep=1},
    every edge/.style={draw=black},
    every label/.style={font=\scriptsize,draw=none,label distance=0}]
    \newcommand{\xa}{2.5}
    \newcommand{\ytop}{.85}

    \node[label={[label distance=-2pt]105:\((1,1)\)}] (c1) at (-\xa,.75*\ytop)
	   {\recipientOf{\snthLiteralOf{1}{i}}};
    \node[label={[label distance=-2pt]105:\((1,1)\)}] (c2) at (0,.75*\ytop)
	   {\recipientOf{\snthLiteralOf{2}{i}}};
    \node[label={[label distance=-2pt]105:\((1,1)\)}] (c3) at (\xa,.75*\ytop)
	   {\recipientOf{\snthLiteralOf{3}{i}}};

    \node[label={[label distance=-7pt]20:\((1)\)}] (vc1) at (-\xa,2*\ytop) {\donorOf{\snthLiteralOf{1}{i}}};
    \node[label={[label distance=-7pt]20:\((1)\)}] (vc2) at (0,2*\ytop) {\donorOf{\snthLiteralOf{2}{i}}};
    \node[label={[label distance=-7pt]20:\((1)\)}] (vc3) at (\xa,2*\ytop) {\donorOf{\snthLiteralOf{3}{i}}};

    \node[label={182:\((2)\)}] (c) at (0,\ytop-1) {$\agentGen(\sClause_i)$};

    \draw[->,arrows={-latex}]
      (c1) edge (vc1)
      (c2) edge (vc2)
      (c3) edge (vc3)
      (c) edge (c1)
      (c) edge (c2)
      (c) edge (c3);
  \end{tikzpicture}%
 }
  \caption{The gadgets in the construction in the proof
   of~Theorem~\ref{thm:decreasing_by_one_nph}. The utility function for all agents is
   the same. The envious agents are framed. Numbers in brackets indicate an initial
   allocation (e.g.,\,$(0)$ denotes an empty bundle; $(1,2)$ means a bundle of
   a value-one and a value-two resource).}
  \label{fig:cnf-gadgets}
 \end{figure}
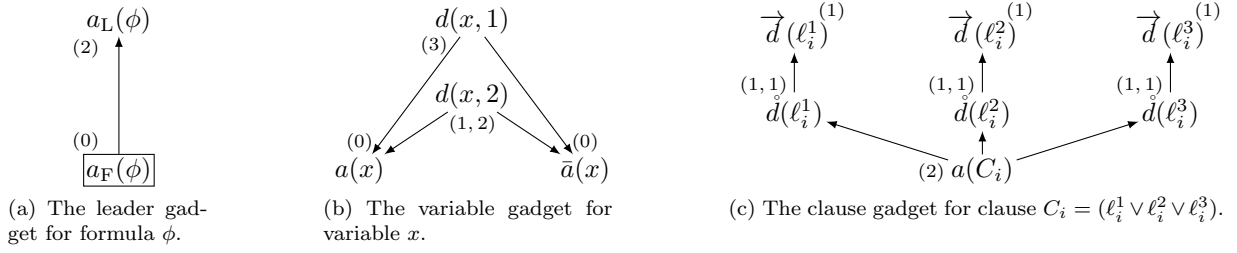

 \paragraph{Construction.}
 The leader gadget consists of two agents:
 the~\emph{leader}~\leaderOf{\sFormula} and
 the~\emph{follower}~\followerOf{\sFormula}. The follower is not assigned any
 resource, while the leader has a resource of value~two. There is a directed arc
 from the follower to the leader. Thus, initially, the follower envies the
 leader.

 The variable gadget for a variable~$\sVariable{} \in \sVariables$ consists of two
 \emph{value agents}~\posAgentOf{\sVariable} and~\negAgentOf{\sVariable}, 
 and two~\emph{dummy agents} \nthDummyOf{1}{\sVariable} and~\nthDummyOf{2}{\sVariable}. The two value agents represent, respectively, assigning ``true'' and ``false''
 to~\sVariable{}. 
 Both dummy agents are paying attention to both value agents (but not vice versa). Initially, three resources are allocated. The resource of value~three is allocated
 to~\nthDummyOf{1}{\sVariable} and two resources, one of value~one and one of
 value~two, are allocated to~\nthDummyOf{2}{\sVariable}. As a result, initially,
 no agent within the gadget envies.

 To describe the clause gadget, let us fix a clause~$\sClause_i =
 (\snthLiteralOf{1}{i} \wedge \snthLiteralOf{2}{i} \wedge
 \snthLiteralOf{3}{i})$. For each literal~\sLiteral{} in the clause we add two
 agents: the \emph{donor}~\donorOf{\sLiteral} and the
 \emph{recipient}~\recipientOf{\sLiteral}. The donor initially has a single
 one-valued resource, while the recipient initially gets two one-valued
 resources. As for the attention relation, for each literal~\sLiteral{}, there is an
 arc~$(\recipientOf{\sLiteral}, \donorOf{\sLiteral})$, that is, the arc points
 from the recipient to the donor. Eventually, the
 gadget contains the~\emph{root} agent~$\agentGen(\sClause_i)$ that gets a
 single resource of value~two. The root agent pays attention to
 recipients~\recipientOf{\snthLiteralOf{1}{i}},
 \recipientOf{\snthLiteralOf{2}{i}}, and \recipientOf{\snthLiteralOf{3}{i}}.
 Observe that, initially, no agent within this gadget envies.

 \begin{figure}[t]
   \centering%
   \begin{tikzpicture}[rotate=90,
    every node/.style={inner sep=1},
    every edge/.style={draw=black,font=\scriptsize},
    every label/.style={font=\scriptsize,draw=none,label distance=0}]
    \newcommand\xend{3}
    \node[label={above:\((3)\)}] (dx) at (.5*\xend,1.5) {\nthDummyOf{1}{\sVariable}};
    \node[label={93:\((1,2)\)}] (vx) at (.5*\xend,0.25)
    {\nthDummyOf{2}{\sVariable}};
    \node[label={left:\((0)\)}] (xx) at (0.2,0) {\posAgentOf{\sVariable}};
    \node[label={left:\((0)\)}] (nx) at (\xend,0) {\negAgentOf{\sVariable}};

    \node[label={135:\((2)\)}] (j) at (0.9*\xend,3) {\leaderOf{\sFormula}};
    \node[label={135:\((0)\)}, draw] (i) at (.5*\xend,3) {\followerOf{\sFormula}};
    \draw[->,arrows={-latex}]  (i) edge (j);
    
    \draw[->,arrows={-latex}]
      (xx) edge (i)
      (nx) edge (i)
      (vx) edge (xx)
      (dx) edge (xx)
      (vx) edge (nx)
      (dx) edge (nx);

    \newcommand{\xa}{.5*\xend}
    \newcommand{\xb}{2.9}
    \newcommand{\xdelta}{1.6}
    \newcommand{\xdots}{2}
    \newcommand{\clauseY}{-1}

    \node[label={87.5:\((1,1)\)}] (c1) at (.5*\xend,\clauseY)
    {\recipientOf{\snthLiteralOf{1}{i}}};

    \node[label={145:\((1)\)}] (vc1) at (\xa-1.25,\clauseY-1) {\donorOf{\snthLiteralOf{1}{i}}};
    \node[label={below:\((2)\)}] (c) at (.5*\xend,\clauseY-1.5) {$\agentGen(\sClause_i)$};

    \draw[->,arrows={-latex}]
      (c1) edge (nx)
      (c1) edge (vc1)
      (c) edge (c1);
  \end{tikzpicture}%
  \caption{Connections of
   gadgets for literal~\snthLiteralOf{1}{i} (over variable~\sVariable{}) of some
   clause~$\sClause_i$. Envious agents are framed. Other possible literals
  of~$\sClause_i$ are omitted for clarity.}
  \label{fig:cnf-full-part}
 \end{figure}
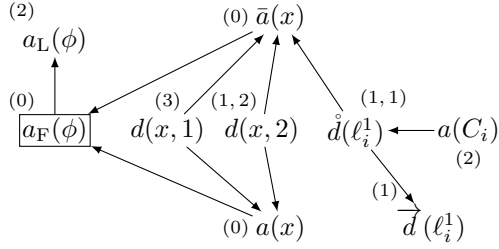

 We obtain the full construction by interconnecting the gadgets (see Figure~\ref{fig:cnf-full-part}).
 First, we connect every value agent with an arc directed from the value agent to the
 follower. Note that we do not introduce any envy because the follower has no
 resources initially. Finally, for each literal~$\snthLiteralOf{j}{i}$, $j \in
 \{1, 2, 3\}$, $i \in [\sClausesCount]$, we add an arc from a respective
 recipient~\recipientOf{\snthLiteralOf{j}{i}} to the value agent of the
 corresponding variable agent; for example, for~$\snthLiteralOf{j}{i} =
 \neg\sVariable_y$, $y \in [\sVariablesCount]$, we would add an
 arc~($\recipientOf{\snthLiteralOf{j}{i}},\negAgentOf{\sVariable_y}$).
 The sharing graph is just the same as the underlying graph of the attention graph.

 The constructed attention graph, the sharing graph, the introduced agents, the resources allocated
 by the initial allocation, and their utilities together with the desired
 number~$\envyThreshold :=0$ of envious agents form an instance
 of~\envyReducingProblemShort{}, clearly computable in polynomial time.

 \paragraph{Correctness.}
 We show that~\sFormula{} is satisfiable if and only if there exists a simple 
 $2$-sharing for the above instance such that no agent will be envious.

 $\Leftarrow :$
 Suppose there exists a simple 
 $2$-sharing for the above instance such that no agent will be envious.
 Since none of value agents has any resource, the only way
 to make a follower non-envious is to share the leader's resource. This %
 makes all value agents envious because they look at the follower that has
 obtained a resource of value~two. Hence, to decrease the number of
 envious agents to zero, one has to actually provide each value agent with a
 resource of value at least two.
 Since respective recipients that are connected to the value agents have 
 only one-valued resources, value agents have to share the resources of the dummy agents. 
 Among these only two suitable resources exist, two- and three-valued ones.
 Thus, in each variable gadget, one of the value agents needs 
 to get a resource of value two and the other a resource of value three. 

 Next, we consider the clause gadget. As was already discussed, each
 value agent gets either a two-valued or a three-valued resource (as a result of
 a sharing that eliminates envy). Notably, every value agent with a three-valued
 resource makes every recipient agent paying attention on the value agent
 envious. Let us fix some literal~$\sLiteral \in \sClause$ whose value
 agent~$\posAgentOf{\sVariable}$ got a three-valued resource. One
 can easily verify that the only way to fix envy
 of~\recipientOf{\sLiteral} is to share a resource from the 
 donor~\donorOf{\sLiteral}. However, then the root agent~$\agentGen(\sClause)$
 becomes envious. A major observation is that this envy
 can only be fixed if there exists a recipient~\recipientOf{\sLiteral'} of another literal of
 clause~\sClause{} that does not need to share with its donor~\donorOf{\sLiteral'},
 which implies that the corresponding value agent for recipient~\recipientOf{\sLiteral'} got a two-valued resource.  As a result, for each clause, at
 least one of the corresponding value agents has to get a two-valued resource.
 Then, the truth assignment that sets~\sVariable{} true if and only if 
 agent~\posAgentOf{\sVariable} is shared with a two-value resource can satisfies every clause of~\sFormula{}.
 
 $\Rightarrow :$
 Suppose there exists a satisfiable truth assignment~\truthAssignment{} of~\sFormula{}.
 We first make the follower share with the leader, making the follower non-envious.
 Next, in each variable gadget, for every variable~\sVariable{} setting to ``True'', we let agent~\posAgentOf{\sVariable} get a two-value resource and \negAgentOf{\sVariable}
 get a three-value resource through sharing with the corresponding dummy agents; otherwise,
 we let agent~\posAgentOf{\sVariable} get a three-value resource and \negAgentOf{\sVariable} get a two-value resource. 
 Finally, for each clause~\sClause{}, since~\truthAssignment{} is a satisfiable truth assignment, there must be a literal~$\sLiteral \in \sClause$ that is assigned ``True'',
 which means the corresponding value agent is shared with a two-value resource.
 Then in the corresponding clause gadget, recipient~\recipientOf{\sLiteral} can share with the root agent~$\agentGen(\sClause)$ such that~$\agentGen(\sClause)$ has value three and the other two recipients can share with their donors such that they have value three.
 It is easy to verify that no agent will be envious under this simple $2$-sharing. 
 \let\sFormula\undefined
 \let\sVariables\undefined
 \let\sVariable\undefined
 \let\sClause\undefined
 \let\sLiteral\undefined
 \let\snthLiteralOf\undefined
 \let\sClausesCount\undefined
 \let\sVariablesCount\undefined
 \let\agentGen\undefined
 \let\leaderOf\undefined
 \let\followerOf\undefined
 \let\posAgentOf\undefined
 \let\negAgentOf\undefined
 \let\nthDummyOf\undefined
 \let\donorOf\undefined
 \let\recipientOf\undefined
 \let\envyThreshold\undefined
 \let\truthAssignment\undefined
\end{proof}

\subsection{Reducing Envy for Tree-like Graphs}

We study how the tree-like structure of the sharing graph and the attention
graph influences the computational complexity of \envyReducingProblemShort.
Studying tree-likeness, we hope for tractability for quasi-hi\-er\-ar\-chi\-cal social
networks, where agents at the same level of the hierarchy influence each other
but they rather do not do so in a cross-hierarchical manner.

Theorem~\ref{thm:hard_clique} shows that when both graphs are (bidirectional) cliques \envyReducingProblemShort{} is \nphard.
We continue to focus on the case when the underlying graph of the attention graph is the same as the sharing graph.
Note that this restriction appears naturally when assuming that one may envy everybody one knows and one may share only with known people. %
Theorem~\ref{thm:XP_treewidth} shows that in this case, if the sharing
graph is a path, a tree or being very close to a tree
(corresponding to a ``hierarchical network''),
then we can solve~\envyReducingProblemShort{} in polynomial time,
while,
 intuitively, %
for sharing graphs %
being ``far from a path,'' presumably there is no algorithm whose
exponential growth in the running time depends only on the ``distance from
path''.

\begin{theorem}
 \label{thm:XP_treewidth}
 When the underlying graph of the attention graph is the same as the sharing graph,
 \envyReducingProblemShort{} can be solved in polynomial time if the sharing graph has a constant treewidth (assuming the tree decomposition is given), 
  and \envyReducingProblemShort{} is \wonehard{} with respect to the pathwidth of the sharing graph.
\end{theorem}

\begin{proof}[Proof of the first part]
Denote by~$tw$ the treewidth of the underlying graph.
Let~$\mathcal{T}=(T,\{X_t\}_{t \in V(T)})$ be a nice tree decomposition (see \cite{DBLP:books/sp/Kloks94} for the definition) of the underlying graph that has width~$tw$.
Before defining the target values for each bag to be computed, 
we first define some concepts and notations.
A \emph{bundle configuration}~$b_t$ for bag~$X_t$ is a function which maps every agent in~$X_t$ to a bundle of resources.
We say a sharing~$\delta$ \emph{realizes} a bundle configuration~$b_t$ if~$\Pi^\delta(a)=b_t(a)$ for each~$a \in X_t$.
A bundle configuration~$b_t$ is \emph{feasible} if there exists a simple $2$-sharing~$\delta$ which can realize~$b_t$.
For a node~$t$ of~$T$, let~$B_t$ be the set of all feasible bundle configurations for bag~$X_t$, 
and let~$V_t$ be the union of all the bags in the subtree of~$T$ rooted at~$t$ (including~$X_t$ itself).
Denote by~$\Env_{[V_t]}(\delta)$ the set of envious agents in the sub-instance induced by agents in~$V_t$ under the sharing~$\delta$.
Note that~$\Env_{[V_t]}(\delta)$ does not contain an agent~$a$ if~$a$ only envies agents outside~$V_t$ under~$\delta$.
For a bundle configuration~$b_t$ and a bag~$X_{t'} \subseteq X_t$, denote by~$b_t[X_{t'}]$ the bundle configuration of~$b_t$ restricted on~$X_{t'}$.

Now we define our target values.
For every node~$t$, every bundle configuration~$b_t^i \in B_t$ and every subset~$S_t \subseteq X_t$, define the following value:
\begin{align*}
f[t,b_t^i,S_t] \coloneqq & \min|\{\Env_{[V_t]}(\delta) \mid \delta \text{ realizes } b_t^i, \Env_{[V_t]}(\delta) \cap X_t =S_t\}|.
\end{align*}
That is, $f[t,b_t^i,S_t]$ is the minimum number of envious agents in the sub-instance induced by~$V_t$ under a sharing~$\delta$ which realizes~$b_t^i$ and~$\Env_{[V_t]}(\delta) \cap X_t =S_t$.
If no such sharing exists, then~$f[t,b_t^i,S_t]=+\infty$.
It is easy to see that our goal is to compute~$\min \{f[r,b_r^i,S_r] \mid b_r^i \in B_r,S_r \subseteq X_r\}$, where~$r$ is the root node.
Next we show how to compute~$f[t,b_t^i,S_t]$ for all~$t$, $b_t^i$ and~$S_t$.

\textbf{Introduce node.} 
Suppose~$t$ is an introduce node with child~$t'$ such that~$X_t=X_{t'} \cup \{v\}$ for some~$v \not \in X_{t'}$.
For any~$b_t^i \in B_t$ and~$S_t \subseteq X_t$, we first check whether~$b_t^i$ and~$S_t$ are compatible.
To this end, we compute a set~$S_t^0 \subseteq X_t$ of agents who envies agents in~$X_t$ under~$b_t^i$.
If~$S_t^0 \setminus S_t \not = \emptyset$, then~$b_t^i$ and~$S_t$ are not compatible, and we set~$f[t,b_t^i,S_t]=+\infty$.
Since all neighbors of~$v$ in~$V_t$ are in~$X_t$, if~$v \not \in S_t^0$, then~$v$ will not be envious in the sub-instance induced by~$V_t$ under any sharing which can realize~$b_t^i$. 
Thus, if~$v \in S_t \setminus S_t^0$, then~$b_t^i$ and~$S_t$ are not compatible, and we set~$f[t,b_t^i,S_t]=+\infty$.
Now for any~$b_t^i \in B_t$ and~$S_t \subseteq X_t$ with~$S_t^0 \subseteq S_t$ and~$v \in S_t^0 \Leftrightarrow v \in S_t$, 
we have that:
\begin{align*}
f[t,b_t^i,S_t]=&\min\{f[t',b_t^i[X_{t'}],S_{t'}]+ |S_t \setminus S_{t'}| \mid S_t \setminus S_t^0 \subseteq S_{t'} \subseteq S_t\}.
\end{align*}

\textbf{Forget node.} 
Suppose~$t$ is a forget node with child~$t'$ such that~$X_t=X_{t'} \setminus \{w\}$ for some~$w \in X_{t'}$.
For any~$b_t^i \in B_t$ and~$S_t \subseteq X_t$, we have that:
\begin{align*}
f[t,b_t^i,S_t]= &\min\{f[t',b_{t'}^j,S_{t'}]  \mid b_{t'}^j[X_t]=b_t^i,S_{t'} \cap X_t=S_t\}.
\end{align*}

\textbf{Join node.} 
Suppose~$t$ is a join node with children~$t_1$ and~$t_2$ such that~$X_t=X_{t_1}=X_{t_2}$
.
For any~$b_t^i \in B_t$ and~$S_t \subseteq X_t$, we have that:
\[
f[t,b_t^i,S_t]= f[t_1,b_t^i,S_t]+f[t_2,b_t^i,S_t]-|S_t|.
\]
Note that~$V_{t_1} \cap V_{t_2}=X_t$ and~$S_t$ is the set of envious agents that have been counted in both~$f[t_1,b_t^i,S_t]$ and~$f[t_2,b_t^i,S_t]$.

Overall, for each node~$t$, we have~$|B_t| \le m^{tw+1}$ since in any bundle configuration for~$X_t$, each agent in~$X_t$ could get at most one additional resource compared to its initial bundle.
So we have at most~$n(2m)^{tw+1}$ values to be computed.
Each value of an introduce node can be computed in~$O(2^{tw+1})$ time.
Each value of an forget node can be computed in~$O(m)$ time.
Each value of an forget node can be computed in constant time.
So the running time is~$O(n(4m)^{tw+2})$.
\end{proof}

\begin{proof}[Proof of the second part]
 \newcommand{\selectorOf}[1]{\ensuremath{s_{#1}}}
 \newcommand{\dummyOf}[1]{\ensuremath{d_{#1}}}
 \newcommand{\providerOf}[1]{\ensuremath{p_{#1}}}
 \newcommand{\resourceOf}[1]{\ensuremath{r({#1})}}
 \newcommand{\resourceOfTwo}[2]{\ensuremath{r({#1},{#2})}}
 \newcommand{\certificationOf}[2]{\ensuremath{c_{#1}^{#2}}}

 We present a parameterized reduction from the \wonehard{} problem \MCC{}; here,
 for an integer~$\ell$, and an undirected graph~$G=(V,E)$ in which each vertex is colored with
 one of~$\ell$ colors, the goal is to find a set of~$\ell$ pairwise adjacent,
 differently colored vertices. Without loss of generality, we assume that
 set~$V$ can be partitioned into~$\ell$ size-$n$
 sets~\namedorderedsetof{V_i}{v^i}{n}, $i \in [\ell]$, where
 each~$V_i$ consists of vertices of color~$i$. Similarly, we assume that there
 are no edges between vertices of the same color. Given an instance~$(G,\ell)$ of \MCC, we construct an instance of
 \envyReducingProblemShort{} as follows (see also Figure~\ref{fig:w1hard_pathwidth}).

To conveniently present the construction of the instance
of~\envyReducingProblemShort{}, we introduce some handy notation. Let us fix a pair
of distinct colors~$i$ and~$j$. We refer to the set of edges connecting the
vertices of these colors as~$E_{i,j}=\{ \{v,v'\} \in E \mid v \in V_i, v' \in V_j\}$. Complementarily, 
let~$\overline E_{i,j}=\{ \{v,v'\} \mid v \in V_i, v' \in V_j\} \setminus E_{i,j}$ and 
$\overline E=\cup_{i \neq j \in [\ell]} \overline E_{i,j}$.
We say that an edge~$e$ is~\emph{forbidden} if~$e \in \overline E$. Indeed, since forbidden edges are not part of~$G$,
they cannot appear between vertices of any clique in~$G$.

\paragraph{Construction.} Our construction consists of the~\emph{vertex
selection} gadget and the~\emph{certification gadget}. We describe in the given
order specifying how they relate to each other in the description of the
certification gadget. For better understanding, we refer to~Figure~\ref{fig:w1hard_pathwidth} showing the construction and Table~\ref{tab:pathwidth_utility} showing the utility functions.

For each color~$i$, we build the~\emph{vertex selection} gadget consisting of
three agents: a \emph{selector}~\selectorOf{i}, a
\emph{provider}~\providerOf{i}, and a \emph{dummy}~\dummyOf{i}. As for the
attention relation, the provider~\providerOf{i} attends the selector~\selectorOf{i}, that, in turn, attends
the dummy~\dummyOf{i}. For each vertex~$v \in V_i$ we create a \emph{vertex}
resource~\resourceOf{v} and give all of them to the provider. The selector
initially gets two resources~\resourceOfTwo{\selectorOf{i}}{1}
and~\resourceOfTwo{\selectorOf{i}}{2} as well as the dummy that
gets~\resourceOfTwo{\dummyOf{i}}{1} and~\resourceOfTwo{\dummyOf{i}}{2}. The
desired goal of this gadget is that the provider shares with the selector
exactly one vertex resource. We get this by making the selector initially
envious of the dummy and ensuring that the dummy cannot share with the selector
to remove the envy. To implement such a behavior, we set the utility function of
the dummy and the provider to be~$0$ for all resources. The selector, however,
gives utility~$1$ to both~\resourceOfTwo{\dummyOf{i}}{1}
and~\resourceOfTwo{\dummyOf{i}}{2} and utility~$2$ to all vertex resources; the
selector gives utility~$0$ to every other resource.
Thus, under the initial allocation, the selector has utility 0 and envies the dummy
whose bundle has utility 2 for the selector.

For each forbidden edge, we build the~\emph{certification gadget}, the purpose of
which is to forbid selectors from being shared vertex resources representing
vertices that are not adjacent in~$G$. To demonstrate the gadget, let us fix a
pair~$\{i, j\}$ of distinct colors such that~$i<j$ and a forbidden 
edge~$e=\{v,v'\} \in \overline E_{i,j}$. The certification gadget
consists of two~\emph{certification} agents~\certificationOf{e}{i}
and~\certificationOf{e}{j}, each having a
resource---\resourceOf{\certificationOf{e}{i}}
and~\resourceOf{\certificationOf{e}{j}}, respectively.
Agent~\certificationOf{e}{j} attends~\certificationOf{e}{i}. We connect the
certification gadget with the respective vertex selection gadgets by letting
both certification agents attend their respective selectors, that
is~\certificationOf{e}{i} attends~\selectorOf{i} and~\certificationOf{e}{j}
attends~\selectorOf{j}. We make two certification agents being envious of the
respective selectors according to the initial allocation, and ensure that unless
the selector are shared nonadjacent (with respect to~$G$) vertex resources, then we can
make one of the certification agents unenvious. Additionally, we make sure that
sharing a resource from a selector to a certification agent cannot make the 
certification agent unenvious. To this end, \certificationOf{e}{i} gives utility~$1$ to both
resources initially given to~\selectorOf{i}. Further, \certificationOf{e}{i}
gives utility~$2$ for vertex resource~\resourceOf{v} and utility~$1$ for all
other vertex resources of color~$i$. We set the utility function
of~\certificationOf{e}{j} for the initial resources of~\selectorOf{j} and vertex
resources of color~$j$ analogously. Finally, both certification agents
give each other's resource utility~$3$ and~utility~$0$ for every other
resource.

Finally we set~$k=\binom{s}{2} n^2-|E|$, which is the number of all forbidden
edges. This finishes our construction of the instance of
\envyReducingProblemShort{}.
Concerning the pathwidth, denote the set of all selectors, all providers and all dummies as~$B$.
There is a path decomposition~$\mathcal{P}=(X_1,X_2,\dots,X_m)$, where~$m=|\overline E|$ is the number of forbidden edges, and each bag~$X_i$ contains all vertices in~$B$ and two certification agents for some forbidden edge.
Therefore, the pathwidth of the underlying graph of the attention graph~$\attentionGraph{}$ is at most~$|B|+2=3\ell+2$.

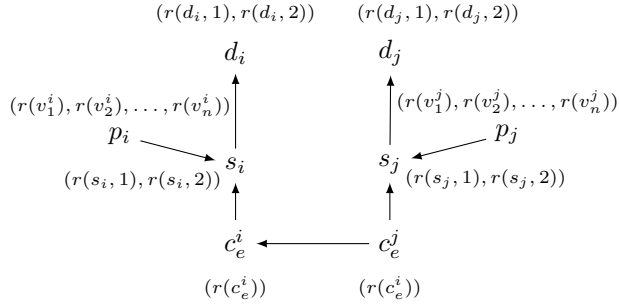
\begin{figure}[t]
   \centering%
   \begin{tikzpicture}
    \node (cei) at (0,0) {\certificationOf{e}{i}};
   \node [right =1.5cm of cei] (cej) {\certificationOf{e}{j}};
   \node [above =.5cm of cei] (ai) {\selectorOf{i}};
   \node [above =.5cm of cej] (aj) {\selectorOf{j}};
   \node [above =of ai] (di) {\dummyOf{i}};
   \node [right =1.5cm of di] (dj) {\dummyOf{j}};
   \node [left = 1cm of ai, yshift=.4cm] (pi) {\providerOf{i}};
   \node [right =1cm of aj, yshift=.4cm] (pj) {\providerOf{j}};

   \node [below =0cm of cei]
   {\scriptsize($\resourceOf{\certificationOf{e}{i}}$)};
   \node [below =0cm of cej] {\scriptsize($\resourceOf{\certificationOf{e}{i}}$)};
   \node [left =-.2cm of ai, yshift=-.2cm] {\scriptsize$(\resourceOfTwo{\selectorOf{i}}{1},
   \resourceOfTwo{\selectorOf{i}}{2})$};
   \node [right =-.2cm of aj, yshift=-.2cm] {\scriptsize$(\resourceOfTwo{\selectorOf{j}}{1},
   \resourceOfTwo{\selectorOf{j}}{2})$};
   \node [above =-.1cm of pi]
   {\scriptsize($\resourceOf{v^i_1},\resourceOf{v^i_2},\dots,
   \resourceOf{v^i_n}$)};
   \node [above =-.1cm of pj] {\scriptsize($\resourceOf{v^j_1},\resourceOf{v^j_2},\dots,
   \resourceOf{v^j_n}$)};
   \node [above =0cm of di] (di5) {\scriptsize($\resourceOfTwo{\dummyOf{i}}{1},
   \resourceOfTwo{\dummyOf{i}}{2}$)};
   \node [right =0.3cm of di5] {\scriptsize($\resourceOfTwo{\dummyOf{j}}{1},
   \resourceOfTwo{\dummyOf{j}}{2}$)};

   \draw[->,arrows={-latex}]
      (cej) edge (cei)
      (cei) edge (ai)
      (cej) edge (aj)
      (pi) edge (ai)
      (pj) edge (aj)
      (ai) edge (di)
      (aj) edge (dj);

  \end{tikzpicture}%
  \caption{The vertex selection gadget for colors~$i$ and~$j$ and the certification gadget for
  a forbidden edge~$e$ with endpoints of colors~$i$ and~$j$ in the constructed
  instance.}
  \label{fig:w1hard_pathwidth}
\end{figure}

\begin{table*}[t]
 \centering
 \begin{tabular}{m{0.5cm}|ccccc}
  \toprule
  $u$ &  $r(v)$ & $r(d_h,t)$ & $r(s_h,t)$ & $r(c_{e'}^h)$
  \\ \midrule
  $d_i$ & 0 & 0 & 0 & 0\\
  $p_i$ & 0 & 0 & 0 & 0\\
  $s_i$ & 2 & 1 & 0 & 0 \\
  $c_e^i$ & 
$ \begin{cases}
      2 & v \in e \cap V_i\\
      1        & \text{o.w.}
  \end{cases}$

 & 0 & 1 & 
 $ \begin{cases}
      3 & e=e' \land i \neq h\\
      0        & \text{o.w.}
  \end{cases}$ \\
 \bottomrule
 \end{tabular}
 \caption{Table of utility functions, where~$t \in \{1,2\}$, $e$ and~$e'$ are forbidden edges, and~$i,h \in [\ell]$. }
 \label{tab:pathwidth_utility}
\end{table*}

\paragraph{Correctness.}
We show that the instance of \MCC{} is a ``yes''-instance if and only if the
constructed instance~of~\envyReducingProblemShort{} is a ``yes''-instance.
Note that under the initial allocation, every selector~$s_i$ envies the
corresponding dummy~$d_i$ since~$u_{s_i}(\pi(s_i))=0 < u_{s_i}(\pi(d_i))=2$. For
each forbidden edge~$e$, both certification agents~$c_e^i$ and~$c_e^j$ are
envious because~$u_{c_e^i}(\pi(c_e^i))=u_{c_e^j}(\pi(c_e^j))=0 <
u_{c_e^i}(\pi(s_i))=u_{c_e^j}(\pi(s_j))=2$.

$\Rightarrow :$ Suppose there is a clique~$C^*$ of size~$\ell$ in~$G$ containing
one vertex from each color.
Based on~$C^*$, we construct a sharing function~$\delta$ such that after this
sharing, all selectors are not envious, and for each forbidden edge~$e$, exactly
one agent from~$\{c_e^i,c_e^j\}$ is envious.
For each color~$i$, let~$v_{q}^i \in C^*$. Then in sharing~$\delta$ let
provider~$p_i$ share resource~$r(v_{q}^i)$ to selector~$s_i$.
Hence, each selector~$s_i$ receives a resource of utility 2 and they become not
envious.
Then for each forbidden edge~$e=\{v_x^i,v_y^j\}$, since~$C^*$ is a clique, we have that~$\{v_x^i,v_y^j\} \not \in C^*$.
According to our current sharing function, the following two events could not happen at the same time:
\begin{enumerate}
\item resource~$r(v_x^i)$ is shared to selector~$s_i$;
\item resource~$r(v_y^j)$ is shared to selector~$s_j$.
\end{enumerate} 
Without loss of generality, assume that resource~$r(v_x^i)$ is shared to selector~$s_i$ and 
resource~~$r(v_{y'}^j)$ instead of~$r(v_y^j)$ is shared to selector~$s_j$.
Then let agent~$c_e^i$ share resource~$r(c_e^i)$ to agent~$c_e^j$ and this makes agent~$c_e^j$ not envious, since~$u_{c_e^j}(\pi(c_e^j) \cup \{r(c_e^i)\})=3=u_{c_e^j}(\pi(s_j) \cup \{r(v_{y'}^j)\})$.
This finishes our definition of sharing~$\delta$.
After this sharing, exactly one certification agent for each forbidden edge is envious and no other agent is envious,
thus there are~$k=\binom{s}{2} n^2-|E|$ envious agents after sharing~$\delta$.

$\Leftarrow :$ Suppose there is a sharing~$\delta$ such that only~$k$ agents will be envious after this sharing.
First, for each forbidden edge~$e$, the only way to make agent~$c_e^i$ not envious is to share resource~$r(c_e^j)$ from~$c_e^j$ to~$c_e^i$ and the only way to make agent~$c_e^j$ not envious is to share resource~$r(c_e^i)$ from~$c_e^i$ to~$c_e^j$.
Since there could be at most one sharing between~$c_e^i$ and~$c_e^j$, at least one agent of~$c_e^i$ and~$c_e^j$ will be envious.
Since~$k=\binom{s}{2} n^2-|E|$ is exactly the number of all forbidden edges, 
it must be that after sharing~$\delta$, no selector~$a_i$ is not envious, and for each forbidden edge~$e$, exactly one agent from~$\{c_e^i,c_e^j\}$ is envious.
Based on this fact, we show in the following that there is a clique~$C^*$ of size~$\ell$ containing one vertex from each color.

To make selector~$s_i$ not envious, $s_i$ has to receive a resource from the provider $p_i$.
For each~$i \in [\ell]$, let~$r(v_{x_i}^i)$ be the resource shared to~$s_i$,
we claim that~$C^*=\{v_{x_1}^1, v_{x_2}^2, \dots, v_{x_\ell}^\ell\}$ is the desired clique.
Suppose~$C^*$ is not a clique, then there are two vertices~$v_{x_i}^i$ and~$v_{x_j}^j$ such that~$\{v_{x_i}^i,v_{x_j}^j\} \not \in E$, which means that there is a forbidden edge~$e=\{v_{x_i}^i,v_{x_j}^j\}$.
Since every selector~$s_i$ shares with provider~$p_i$, agents~$c_e^i$ and~$c_e^j$ can only share with each other and the maximum utility for them after such a sharing is at most 3.
However, after the sharing, $u_{c_e^i}(\Pi^\delta(s_i))=u_{c_e^i}(\{\pi(s_i),r_{x_i}^i\})=4$ and $u_{c_e^j}(\Pi^\delta(s_j))=u_{c_e^j}(\{\pi(s_j),r_{x_j}^j\})=4$,
then both~$c_e^i$ and~$c_e^j$ will be envious, which is a contradiction.
Therefore, $C^*$ is a clique of size~$s$ containing one vertex from each color.
\end{proof}

\subsection{Reducing Envy for Identical Utility Functions}

We proceed by studying the natural special case where all agents have 
the same utility
function. Already in this constrained setting of homogeneous agents allocation
problems frequently become hard~\citep{BL08}. This is why this scenario also
attracts quite some attention in the fair allocation
literature~\citep{NRR13,BB18,BKV18,BL08}.
Here, we focus on cliques that naturally model small, dense communities and
allow for convenient comparison with the classical setting of indivisible,
non-shareable resources (where the attention graph is implicitly assumed to be a
bidirectional clique).

Theorem~\ref{thm:hard_clique}  already shows that
restricting the attention graph and the sharing graph to be cliques is not
enough to make \envyReducingProblemShort{} solvable in polynomial
time. However, if all agents have the same utility function,
Theorem~\ref{thm:P_clique+same_utility} shows that restricting the attention graph to
be a clique alone can achieve polynomial-time solvability.
The idea behind the proof of
Theorem~\ref{thm:P_clique+same_utility} is that unenvious agents are exactly those with
the highest utility, so we can guess the highest utility and then reduce the problem to a bounded number of \textsc{Maximum Matching}.

\begin{theorem}
\envyReducingProblemShort{} is solvable
in polynomial time if the attention graph is
a bidirectional clique and all agents have the same utility function.
\label{thm:P_clique+same_utility}
\end{theorem}
\begin{proof}
Let~$u$ be the utility function for all agents and let~$N$ be the set of
non-envious agents after a sharing~$\delta$.
For every two agents~$a_i,a_j \in N$, since the attention graph is a clique,
we have that~$u(\sharingAllocation^\delta(a_i)) \ge
u(\sharingAllocation^\delta(a_j))$ and~$u(\sharingAllocation^\delta(a_j)) \ge
u(\sharingAllocation^\delta(a_i))$; hence $u(\sharingAllocation^\delta(a_i)) =
u(\sharingAllocation^\delta(a_j))$. Denote the utility of all agents in~$N$
by~$u^*$.
Clearly, $u(\sharingAllocation^\delta(a_k)) < u^*$ for every agent~$a_k \not \in N$.
Thus, under any sharing allocation, the set of unenvious agents is exactly the set of agents who have the highest utility.
Based on this observation, it suffices to compute, for each possible target utility~$u^*$,
the largest number of agents that can have exactly utility~$u^*$ after some
sharing.

We first show that the number of different possibilities for the target utility~$u^*$ is~$O(nm)$.
Each agent can get at most one resource through a simple $2$-sharing and there are overall~$m$ different resources, so each agent could end with~$O(m)$ different utilities.
Together with the number~$n$ of agents, we get the desired upper bound~$O(nm)$. 
Let~$S$ be the set of all these utilities, that is, $S=\{t \mid t=u(\pi(a_i))+u(r) \text{ for some agent } a_i \in \agents \text{ and some resource } r \in \resources \setminus \pi(a_i)\}$.
Let~$u_0=\max_{i \in \agents} u(\pi(a_i))$ be the largest utility of some agent before a sharing.
We only need to consider utility values in~$S'=\{t \in S \mid t \ge u_0\}$ as
any target utility~$u^*$ should be at least~$u_0$.

Next we show that for every target utility~$u^* \in S'$ computing the largest
number of agents that can have utility~$u^* $ after a sharing~$\delta$ can be reduced
to \textsc{Maximum Matching}. 
We first find the set~$N_0$ of agents who already have utility~$u^*$ before sharing, that is, $N_0=\{a_i \in \agents \mid u(\pi(a_i))=u^*\}$.
Notice that $N_0=\emptyset$ if $u^*>u_0$. 
Then we construct a graph~$G=(\agents,E)$ where vertices being the agents and an edge~$e=\{a_i,a_j\}$ belongs to~$E$ if one of agents~$a_i$ and~$a_j$ can strictly increase its utility to~$u^*$ through a sharing between~$a_i$ and~$a_j$.
Formally, $e=\{a_i,a_j\} \in E$ if~$\{a_i,a_j\} \in \shareRelations$, $a_i \not \in N_0$, and there exists a resource~$r \in \pi(a_j)$ such that~$u(\pi(a_i))+u(r)=u^*$ or $a_j \not \in N_0$ and there exists a resource~$r \in \pi(a_i)$ such that~$u(\pi(a_j))+u(r)=u^*$.
By computing a maximum matching in~$G$, which can be done in polynomial time~\citep{HK73},
 we can find the largest number of agents in~$\agents \setminus N_0$ who can increase their utilities to~$u^*$ through a sharing.
\end{proof}

We complement Theorem~\ref{thm:P_clique+same_utility} by showing in Theorem~\ref{thm:hard_share_clique+same_utility} that identical utility functions together with the sharing graph being a clique
 are not sufficient to make \envyReducingProblemShort{} solvable in
polynomial time. 
Note that Theorem~\ref{thm:P_clique+same_utility} and Theorem~\ref{thm:hard_share_clique+same_utility} show an interesting contrast between the influence of 
the completeness of the attention graph and the sharing graph 
on the problem's computational complexity.

\begin{theorem}
\label{thm:hard_share_clique+same_utility}
 \envyReducingProblemShort{} is \nphard{} even if
 the sharing graph is a clique,
 all agents have the same utility function,
 and the maximum initial bundle size is one.
 For the same constraints, \envyReducingProblemShort{} is~\wonehard{} with
 respect to the parameter ``number of resources.''
\end{theorem}
\begin{proof} 
 \newcommand{\cGraph}{\ensuremath{G}}
 \newcommand{\cEdges}{\ensuremath{E}}
 \newcommand{\cVertices}{\ensuremath{V}}
 \newcommand{\cCliqueSize}{\ensuremath{\ell}}
 \newcommand{\cInstance}{\ensuremath{\bar{I}}}
 \newcommand{\cVertex}{\ensuremath{v}}
 \newcommand{\cVerticesCount}{\ensuremath{n}}
 \newcommand{\cEdge}{\ensuremath{e}}
 \newcommand{\cEdgesCount}{\ensuremath{m}}
 \newcommand{\cCliqueEdgesCount}{\ensuremath{\tilde{\ell}}}

 \newcommand{\instance}{\ensuremath{I}}
 \newcommand{\agentGen}{\ensuremath{a}}
 \newcommand{\dummyAgent}{\ensuremath{d}}
 \newcommand{\envyThreshold}{\ensuremath{k}}

 We give a polynomial-time many-one reduction from \clique{} (that is \nphard{};
 and \wonehard{} with respect to the number of vertices in a sought clique), where we are given an
 undirected graph and an integer~\cCliqueSize{}, and the question is whether there
 is a set of~\cCliqueSize{} mutually connected vertices. To this end, we fix an
 instance~$\cInstance=(\cGraph, \cCliqueSize)$~of \clique{},
 where~\namedorderedsetof{\cVertices}{\cVertex}{\cVerticesCount} and
 \namedorderedsetof{\cEdges}{\cEdge}{\cEdgesCount}. For brevity,
 let~$\cCliqueEdgesCount := {\cCliqueSize \choose 2}$ be the number of edges in
 a clique of size~$\cCliqueSize$. Without loss of generality we assume that~$4
 \leq \cCliqueSize < \cVerticesCount$ and $\cCliqueEdgesCount \le m$. We build
 an instance~$\instance$ of~\envyReducingProblemShort{} corresponding
 to~\cInstance{} as follows.

 For each vertex~$\cVertex_i \in \cVertices$ we create a~\emph{vertex
 agent}~$\agentGen(\cVertex_i)$ and for each edge~$\cEdge_i \in \cEdges$ we create
 an~\emph{edge agent}~$\agentGen(\cEdge_i)$. We also add~$2m$ \emph{dummy agents}~$\{d_1,d_2,\dots,d_{2m}\}$ and \cCliqueEdgesCount{} \emph{happy agents}~$\{h_1,h_2,\dots,h_{\cCliqueEdgesCount}\}$.
 We introduce \cCliqueEdgesCount{}
 resources~\orderedlistingof{\resource}{\cCliqueEdgesCount}.  In the initial
 allocation~\allocation{}, the resources are given to the happy agents such that
 each happy agent gets one resource.
 The resources are
 indistinguishable for the agents, that is, each agent gives all of them the
 same utility value; for convenience, we fix this utility to be one.
 The sharing graph is a clique.
 In the attention graph~\attentionGraph{}, 
 all edge agents have outgoing arcs to all happy agents,
 each vertex agent~$\agentGen(\cVertex)$, $\cVertex \in \cVertices$, has an outgoing arc to every edge agent corresponding to an edge incident to~\cVertex{},
 and all dummy agents have outgoing arcs to all vertex agents, 
 The described ingredients, together with
 the desired number of envious agents at most~$\envyThreshold := \cEdgesCount -
 \cCliqueEdgesCount + \cCliqueSize$, build an
 instance~$\instance$ of~\envyReducingProblemShort{} corresponding to \clique{}
 instance~\cInstance{}. The whole construction can be done in polynomial time.

 Observe that all edge agents are envious under~\allocation{}.
 Since vertex agents and dummy agents do not pay attention to happy agents, no more agents are
 envious. Thus, there are exactly $\cEdgesCount{} > \envyThreshold$~envious agents under~\allocation{}.
 We now show that instance~\cInstance{} of \clique{} is a ``yes''-instance if and only if instance~$\instance$ of~\envyReducingProblemShort{} is a ``yes''-instance.

 $\Rightarrow:$
 Assume that there exists a clique of size~$\cCliqueSize$ in~\cGraph{}. Let us
 call each agent corresponding to an edge in the clique a \emph{clique agent};
 by definition, there are exactly~\cCliqueEdgesCount{} of them. We construct
 a desired sharing for~\instance{} by giving every clique agent a resource.
 Indeed, it is achievable since we have exactly~\cCliqueEdgesCount{} resources,
 which can be shared by the happy agents with exactly
 \cCliqueEdgesCount~distinct clique agents. After such a sharing, all clique
 agents do not envy any more. Observe that the sharing made~\cCliqueSize{}
 agents corresponding to the vertices of the clique envious. Thus, the total
 number of envious agents is decreased by~$\cCliqueEdgesCount{}-\cCliqueSize{}$.
 This gives exactly the desired
 number~\envyThreshold{}, resulting in~\instance{} being a ``yes''-instance.

 $\Leftarrow:$
 Let~$\delta$ be a sharing such that there are at most~$k$ envious agents under the sharing allocation~$\sharingAllocation^\delta$.
 Since the sharing graph is a clique, happy agents could share resources with any one of the remaining agents.
 However, if a happy agent shares with a vertex agent, then all~$2m$ dummy
 agents will envy this vertex agent.
 Since there are only \cCliqueEdgesCount{} resources, which means at least~$2m-\cCliqueEdgesCount \ge m>k$ dummy agents will be envious at the end, which is a contradiction.
 Hence, we can assume that no vertex agent is involved in~$\delta$.
 Consequently, no dummy agent will envy as they only look at vertex agents.
 Hence, without loss of generality, we can additionally assume that no dummy agent is involved in~$\delta$.
 Then all sharings in~$\delta$ are restricted to edge agents and happy agents.

 Now, we show that~$\delta$ actually encodes a clique of size~\cCliqueSize{} in~$G$. 
 When an edge agent gets a resource through sharing, it stops being
 envious, however, all vertex agents that looks at this edge agent
 start becoming envious. 
 So, an ultimate goal for an optimal sharing is to share resources to the highest possible number of edge agents that are connected with
 an outgoing arc with as few vertex agents as possible. 
 Formally, let~$S$ be the set of edge agents who get shared resources in~$\delta$ and~$T$ the set of the corresponding vertex agents who will become envious.
 Denote~$s=|S|$ and~$t=|T|$.
 Naturally, $s \le \binom{t}{2}$ and~$s \le \cCliqueEdgesCount$.
 Then, the number of envious agents after the sharing in the solution is at
 least~$\cEdgesCount - s+t$
and this number should be at most~$k=\cEdgesCount - \cCliqueEdgesCount + \cCliqueSize$, that is,
\[
\cEdgesCount - s+t \le \cEdgesCount - \cCliqueEdgesCount + \cCliqueSize
\Rightarrow s-t \ge \cCliqueEdgesCount - \cCliqueSize.
\]
Since~$\cCliqueEdgesCount = {\cCliqueSize \choose 2}$ and~$s \le \binom{t}{2}$,
if~$t < \cCliqueSize$, then~$s-t \le \binom{t}{2} - t <{\cCliqueSize \choose 2} - \cCliqueSize$,
which contradicts with the above $s-t \ge \cCliqueEdgesCount - \cCliqueSize$. Hence~$t \ge \cCliqueSize$. Since~$s \le
\cCliqueEdgesCount$, the only way to satisfy~$s-t \ge \cCliqueEdgesCount -
\cCliqueSize$ is that~$s=\cCliqueEdgesCount$ and~$t=\cCliqueSize$. Thus, the
corresponding vertices of~$T$ form a clique of size~\cCliqueSize{} in \cGraph.

Observe that we use exactly $\cCliqueEdgesCount$~resources, which proves the
claimed \wonehardness{}.
 \let\cGraph\undefined
 \let\cEdges\undefined
 \let\cVertices\undefined
 \let\cCliqueSize\undefined
 \let\cInstance\undefined
 \let\cVertex\undefined
 \let\cVerticesCount\undefined
 \let\cEdge\undefined
 \let\cEdgesCount\undefined
 \let\cCliqueEdgesCount\undefined
 \let\instance\undefined
 \let\agentGen\undefined
 \let\dummyAgent\undefined
 \let\envyThreshold\undefined
\end{proof}

Finally, we observe that for the scenario with a
constant number of shared resources, there is a na\"{i}ve brute-force algorithm
running in polynomial time.
\begin{observation}%
\label{obs:resources_envy_red_xp}
\envyReducingProblemShort{} is solvable in polynomial time if the number of
shared resources
(or the number of resources) 
is a constant.
\end{observation}
\begin{proof}
 Let~$b\le m$ be the upper bound on the number of shared resources. There are at
 most~$m^b$~different choices of resources to be shared, and each resource has
 at most~$n-1$ possible receivers, so overall we need to check at
 most~$m^bn^b$~cases.
\end{proof}%

\section{Extensions}\label{sec:extension}

We introduce two natural extensions of our model, both of which describe costs
which can be incurred by sharing---either for the central authority or agents.
For each extension, we discuss which results can easily be adapted to cover
them. Finally, we formally present how to modify the proofs of the relevant
results.

\subsection{Extension 1: Loss by Sharing}
So far, we have assumed that agents get the full utility of the shared resources.
This does not hold for situations in which sharing causes more inconvenience than
owning a resource alone. Nonetheless, many of our algorithms can be easily adapted to
deal with (computational) issues that arise in such situations.
Consider the case when agents get only a
fraction of the full utility from shared resources.
Then our algorithms for improving utilitarian welfare (Theorem \ref{thm:UWSA-P}) and egalitarian welfare (Lemma \ref{lem:EWSA-1-P}) still work with minor changes.
For reducing envy it might be that after sharing agents lose some utility and thus become envious.
Again, our algorithms for few agents (Theorem \ref{thm:fpt-agents}) or identical
utility functions (Theorem \ref{thm:P_clique+same_utility}) still work with
minor changes.

Formally, 
we introduce two parameters $\alpha,\beta \in [0,1]$ to quantify the effect that
agents do not get the full utility of the shared resources, that is,
for any resource $r$ initially assigned to agent $a_i$ in $\pi$ and shared to $a_j$ under $\sharing$, the utility of resource $r$ for $a_i$ is $\alpha \cdot u_i(r)$ and that for $a_j$ is $\beta \cdot u_j(r)$.
Recall that we refer to $\sharingAllocation^\sharing(\agent)$ as a \emph{bundle} of~\agent{}.
Now we refer to $\sharingAllocation^\sharing_{+}(\agent)$ as the set of resources shared to \agent{} by other agents, $\sharingAllocation^\sharing_{-}(\agent)$ the set of resources shared to other agents by \agent{}, and $\sharingAllocation^\sharing_{0}(\agent)$ the set of remaining unshared resources.
Then the utility of agent \agent{} under $\sharingAllocation^\sharing$ is 
\[
 \utilityFunction_{i}(\sharingAllocation^\sharing(\agent_i))=\sum_{r \in \sharingAllocation^\sharing_0(\agent)} u_i(r) + \sum_{r \in \sharingAllocation^\sharing_+(\agent)} \alpha u_i(r) +\sum_{r \in \sharingAllocation^\sharing_-(\agent)} \beta u_i(r).
\]

Later in this section, we consider this extension in the proofs of
Theorem~\ref{thm:UWSA-P}, Lemma~\ref{lem:EWSA-1-P},
Theorem~\ref{thm:fpt-agents}, and Theorem~\ref{thm:P_clique+same_utility}.
Notice that the original setting corresponds to the case with $\alpha=\beta=1$.

\subsection{Extension 2: Cost of Sharing}
It is also natural to assume that the central authority would need to pay some cost for
each sharing to incentivize agents to share resources. In this case, there would
be a limited budget that the central authority can spend and the goal would be
to improve the allocation through sharings whose costs do not exceed the budget.
So far, our model was not capable of modeling the described scenario. However,
for this generalized setting, our algorithms for improving egalitarian welfare
(Lemma~\ref{lem:EWSA-1-P}), reducing envy 
for few agents (Theorem~\ref{thm:fpt-agents}) or identical utility functions
(Theorem~\ref{thm:P_clique+same_utility}) still work with minor changes.

More precisely, for a simple $2$-sharing~$\sharing$, we introduce a
\emph{sharing cost} $c_g: \shareRelations \rightarrow \mathbb{N}$ and a budget
$C \in \mathbb{N}$ for the central authority to incentivise sharing.  Notice
that the sharing cost is defined for each pair of agents.  The cost of a
sharing~$\sharing$ for the central authority is
\[
c(\sharing)=\sum_{
\delta(\{a_i,a_j\})\neq \emptyset} c_g(\{a_i,a_j\}).
\]
Now the goal of the central authority is to find a sharing~$\sharing$ with
$c(\sharing)\le B$ to improve the allocation.

Later in this section, we consider this extension in the proofs of
Lemma~\ref{lem:EWSA-1-P}, Theorem~\ref{thm:fpt-agents}, and
Theorem~\ref{thm:P_clique+same_utility}.  Notice that the original setting
corresponds to the case with $B=\sum_{\{a_i,a_j\} \in \shareRelations}
c_g(\{a_i,a_j\})$, i.e., $c(\sharing)\le B$ holds for any sharing $\sharing$.

\subsection{Respective Proofs' Modifications}

In the ensuing paragraphs, we list the necessary modifications of our results to
make them work correctly for the aforementioned extensions. 

\paragraph{Proof of Theorem \ref{thm:UWSA-P} for Extension 1.}
We adapt the proof of Theorem \ref{thm:UWSA-P} given earlier to work for Extension 1.
To this end, we just need to change the weight of edges in the constructed graph as follows.
The weight of the edge between $v_{i_1}^{j_1}$ and $v_{i_2}^{j_2}$ is now defined as
\[
\max\{\beta u_{i_1}(r_{j_2})-(1-\alpha) u_{i_2}(r_{j_2}),\beta u_{i_2}(r_{j_1})-(1-\alpha)u_{i_1}(r_{j_1})\}
\] instead of $\max\{u_{i_1}(r_{j_2}),u_{i_2}(r_{j_1})\}$, i.e., the weight is the increased utilitarian social welfare in the new setting.
Accordingly, in the ``$\Leftarrow:$'' part, for each non-dummy edge $(v_{i_1}^{j_1},v_{i_2}^{j_2}) \in M$, we set $\sharing(a_{i_1},a_{i_2})=r_{j_1}$ if the increased utilitarian social welfare by sharing $r_{j_1}$ to agent $a_{i_2}$ is no smaller than that of sharing $r_{j_2}$ to agent $a_{i_1}$, and $\sharing(a_{i_1},a_{i_2})=r_{j_2}$ otherwise.
Then, using the same arguments, we can show that there is a $b$-bounded $2$-sharing $\sharing{}$ such that $\utilitarianWelfare(\sharingAllocation^\sharing) \geq k$ if and only if there is matching~$M$ in graph~$G$ with weight~$\sum_{e \in M}w(e) \ge k-\utilitarianWelfare(\pi)+P$, and thus, the problem can be reduced to \textsc{Maximum Weighted Matching}.
\qed

\paragraph{Proof of Lemma \ref{lem:EWSA-1-P} for Extensions 1 and 2.}
\newcommand{\utilityThreshold}{\ensuremath{k}}
 \newcommand{\betterAgents}{\ensuremath{\agents^+_\utilityThreshold}}
 \newcommand{\worseAgents}{\ensuremath{\agents^-_\utilityThreshold}}
We adapt the original proof of Lemma~\ref{lem:EWSA-1-P} to work for the case
with both Extension 1 and Extension 2.
Similarly as before, we partition the set~\agents{} of agents into two
 sets \betterAgents{} and~\worseAgents{} containing, respectively, the agents
 with their bundle value under~$\allocation$ at least~\utilityThreshold{} and smaller than~\utilityThreshold{}.
 When constructing the graph~$\graph_\utilityThreshold =
 (\betterAgents, \worseAgents, E_\utilityThreshold)$,
 for two
 agents~$\agent_i \in \betterAgents$ and~$\agent_j \in \worseAgents$ that are
 neighbors in the sharing graph~\shareGraph{},
 we add an edge $e = \{\agent_i, \agent_j\}$ belongs
 to~$E_\utilityThreshold$ if~$\agent_i$ can share a resource with~$\agent_j$ to
 raise the utility of the latter to at least~\utilityThreshold{}
 \emph{and} after this sharing the utility of $\agent_i$ is at least~\utilityThreshold{}; 
 formally,
 there exists a resource~$\resource \in \allocation(\agent_i)$ such
 that~$\utilityFunction_j(\allocation(\agent_j)) +
 \beta \utilityFunction_j(\resource)) \geq \utilityThreshold$
 and $\utilityFunction_i(\allocation(\agent_j)) -
 (1-\alpha) \utilityFunction_i(\resource)) \geq \utilityThreshold$.
 In addition, for each edge~$e=\{a_i,a_j\} \in E_d$ we assign it a weight~$c_g(\{a_i,a_j\})$.
 With similar arguments, we have that
  there is a simple $2$-sharing~$\delta$ with~$c(\delta) \le B$ and~$\egalitarianWelfare(\sharingAllocation^\sharing) \geq k$ if and only if there is matching~$M$ in
 graph~$G_k$ with~$\sum_{e \in M}c_g(e) \le B$ and~$|M| \ge |\worseAgents|$.
 Thus, we just need to check whether there is matching~$M$ in
 graph~$G_k$ with~$\sum_{e \in M}c_g(e) \le B$ and~$|M| \ge |\worseAgents|$, which is an instance of WBMM and is solvable in polynomial time according to Lemma \ref{lem:matching_p}.
 \qed

\paragraph{Proof of Theorem \ref{thm:fpt-agents} for Extension 1 and 2.}
We adapt the original proof of Theorem \ref{thm:fpt-agents} to work for the case
with both Extension 1 and Extension 2.
The extended algorithm is given in Algorithm \ref{alg:fpt-agents-gener}.

For Extension 2 where there is cost for sharing, 
notice that the cost of any realization~$\delta$ of~$M$ is fixed: $c(\delta)=c(M)$.
Therefore, for any guessed~$M$, the extended algorithm checks in the first line whether the cost of the configuration $M$ is at most $B$.

For Extension 1 where agents only get a fraction of utility from shared resources, in the for-loop of Algorithm \ref{alg:fpt-agents-gener}, the set of deleted resources from $P_i^0$ are changed accordingly.
In addition, since agent may lose utility after sharing, it is important to guarantee that target agents in $C$ who will share resources to other agents according to the configuration $M$ does not becomes envious after the sharing.
To guarantee this, 
we partition all agents in $C$ into three subsets $C_+ \cup C_- \cup C_0$, where 
\begin{align*}
C_+&=\{a_i \in C \mid \exists j \text{ such that } (j,i) \in M \}; \\
C_-&=\{a_i \in C \mid \exists j \text{ such that } (i,j) \in M \};\\
C_0&=\{a_i \in C \mid \not \exists j \text{ such that } (i,j) \in M \text{ or } (j,i) \in M\}. \\
\end{align*}
Then we change the definition of $P^0_i$ (defined in Eq. \ref{eq:S_i}) as follows:
\begin{equation}
  P^0_i :=
    \begin{cases*}
      \pi(j) \cup \{d_i\} & if $a_i \in C_+$,\\
      \pi(i) \cup \{d_i\} & if $a_i \in C_-$,\\
      \{d_i\}        & if $a_i \in C_0$.
    \end{cases*}
 \end{equation}
Finally, the forbidden resource defined in Definition \ref{def:blocking_resource} should be changed accordingly, i.e., 
resource~$r \in P_i$ is a~\emph{forbidden resource} for some target
  agent~$a_i$ if there is some target agent~$a_j \in C_+ \cup C_0$ with~$(a_j, a_i) \in
  \attentionGraph$ such that $$\max\{u_j(\pi(j))+ \beta u_i(r') \mid r' \in P_j\}
  < u_j(\pi(i))+ \beta u_j(r),$$
  or there is some target agent~$a_j \in C_-$ with~$(a_j, a_i) \in
  \attentionGraph$ such that $$\max\{u_j(\pi(j))-(1-\alpha) u_i(r') \mid r' \in P_j\}
  < u_j(\pi(i))+ \beta u_j(r).$$
  \qed
 
\paragraph{Proof of Theorem \ref{thm:P_clique+same_utility} for Extensions 1 and 2.}
We adapt the earlier version of the proof to work for the case with both Extension 1 and Extension 2.
For extension 1 where agents only get a fraction of utility for shared resources, if $\alpha=1$, then as no agent will lose utility after sharing, we just need to change the edge set $E$ as follows.
We add $e=\{a_i,a_j\}$ to~$E$ if~$\{a_i,a_j\} \in \shareRelations$, $a_i \not \in N_0$, and there exists a resource~$r \in \pi(a_j)$ such that~$u(\pi(a_i))+\beta u(r)=u^*$ or $a_j \not \in N_0$ and there exists a resource~$r \in \pi(a_i)$ such that~$u(\pi(a_j))+ \beta u(r)=u^*$.
If $\alpha<1$, then in additional to the above change, when constructing the graph~$G=(\agents,E)$, we do not consider those agents who already have utility $u^*$, as they cannot help other agents increase utility without making their utility less than $u^*$.
For Extension 2 where there is a cost for sharing, we need to convert the graph~$G=(\agents,E)$ into a weighted graph by assigning each edge~$e \in E$ a weight~$c_g(e)$.
Then by computing a maximum-cardinality matching with weight at most~$C$ in~$G$, which can be done in polynomial time according to Lemma~\ref{lem:matching_p},
 we can find the largest number of agents in~$\agents \setminus N_0$ who can increase their utilities to~$u^*$ through a sharing.
 \qed

\begin{algorithm}[t]
 \caption{Generalization of Algorithm \ref{alg:fpt-agents} for two extensions.
  \label{alg:fpt-agents-gener}}
\textsl{DoesFeasibleRealizationExist+}$(\attentionGraph,\pi,\{u_i\}_{i \in C},B,C,M)$\\
\lIf{$c(M) > B$}{\Return{``no''}}
\For{each agent~$a_i \in C$}{
 $P_i \leftarrow P^0_i \setminus \{r \in P_i \mid
 u_i(\pi(i)) +\beta u_i(r)< t_i\}$\;
}
\Repeat{$B = \emptyset$}{
 $B \leftarrow \bigcup_{a_i \in C} F_i(\{P_1, P_2,
 \ldots, P_{|C|}\})$\;
 $P_i \leftarrow P_i \setminus B$\;
}
 \lIf{$\exists i$ with $P_i=\emptyset$}{\Return{``no''} \textbf{else return} ``yes''}
\end{algorithm}

\section{Conclusion}\label{sec:concl} 
We brought together two important topics---fair
allocation of resources and resource sharing. 
Already our very basic and simple model where each resource can 
be shared by neighbors in a social network and each agent can participate 
in a bounded number of sharings led to challenging 
computational problems.
We shed light at their fundamental computational complexity
limitations (in the form of computational hardness) and provided
generalizable algorithmic techniques (as mentioned in
Section~\ref{sec:extension}).
Our results are of broader interest in at least two respects. First, we 
gained insight into a recent line of research aiming at achieving
fairness without relaxing its requirements too much.
Second, 
there is rich potential for future research exploring our
general model of sharing allocations (in its full power described
by~Definition~\ref{def:sharing_allocation}).

\paragraph{Beyond 2-sharing.}
We focused on sharing resources between neighbors in a social network. Yet,
there are many scenarios where sharing resources among a large group of agents can be very natural and
wanted. 
If every resource can be shared with everyone, then
there is a trivial envy-free allocation. Hence, it is interesting to further study the
limits of existence of envy-free allocations under various sharing relaxations
(concerning parameters such as number of shared resources, number of agents sharing a resource, etc.).

\paragraph{No initial allocation.}
In our model, the sharing builds up on an initial allocation.
It is interesting to study the case without initial allocation,
i.e., allocating indivisible but shareable resources to achieve welfare and/or
fairness goal.

\paragraph{Combinations of graph classes.}
We showed that \envyReducingProblemShort{} is \nphard{} even if both input graphs 
are (bidirectional) cliques, but it is polynomial-time solvable if two graphs are the same and have constant treewidth.
Based on this, analyzing various combinations of graph classes of the two social networks might be valuable. %

\paragraph{Strategic concerns and robustness.} We have assumed that all utility values are
truthfully reported as well as correct and that the agents need not to be incentivized to
share resources. Neither of these assumptions might be justified in some
cases---the agents might misreport their utility, the utility values might be
slightly incorrect, or a sharing can come at a cost \emph{for agents}
(splitting utility from shared resources, as described in
Section~\ref{sec:extension}, is an example of the latter). Tackling this kind
of issues opens a variety of directions, which includes studying strategic
misreporting of utilities, robustness of computed solutions against small
utility values perturbations, and finding allocations that incentivize sharing.

\section*{Acknowledgments}
We are grateful to the AAAI'22 reviewers for their insightful comments.
This work was started when all authors were with TU~Berlin.
Andrzej Kaczmarczyk was supported by the DFG project “AFFA” (BR 5207/1 and NI
369/15) and by the European Research Council (ERC).
Junjie Luo was supported by the
DFG project “AFFA” (BR 5207/1 and NI 369/15) and the Singapore Ministry of
Education Tier 2 grant (MOE2019-T2-1-045).

This project has received funding from the European Research Council (ERC)
under the European Union’s Horizon 2020 research and innovation programme (grant
agreement No 101002854).

\noindent \includegraphics[width=3cm]{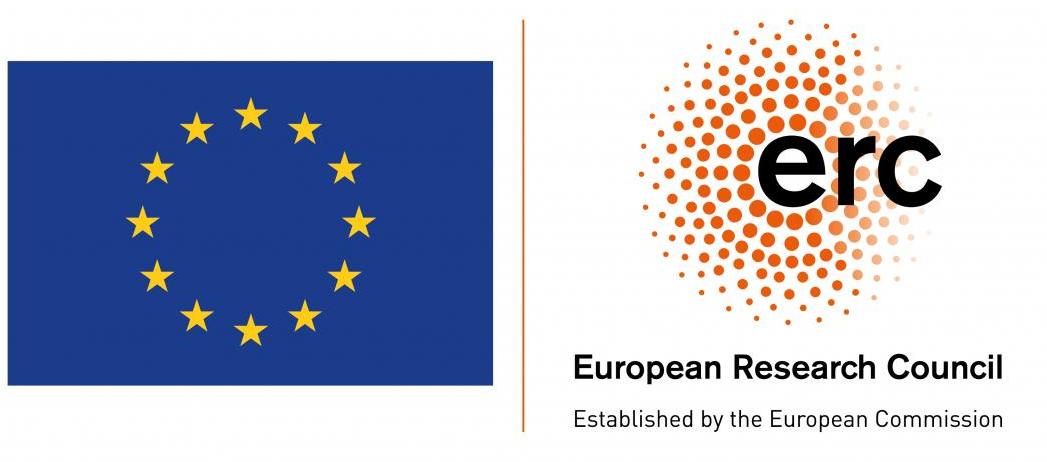}

\bibliographystyle{plainnat}
\bibliography{bibliography.bib}

\end{document}